%% file: main.tex
\documentclass[a4paper,UKenglish,cleveref, autoref, thm-restate]{lipics-v2021}

\pdfoutput=1 
\hideLIPIcs  

\usepackage{proof}
\usepackage{multirow}
\usepackage{listings}
\usepackage{tikz}
\usetikzlibrary{calc}
\usepackage{soul}
\usepackage{xspace}

\input{macros.tex}

\newtheorem{invariant}{Invariant}

\definecolor{dkgreen}{rgb}{0,0.6,0}
\definecolor{ltblue}{rgb}{0,0.4,0.4}
\definecolor{dkviolet}{rgb}{0.3,0,0.5}

\lstdefinelanguage{Coq}{
    mathescape=true,
    texcl=false,
    morekeywords=[1]{Section, Module, End, Require, Import, Export,
        Variable, Variables, Parameter, Parameters, Axiom, Hypothesis,
        Hypotheses, Notation, Local, Tactic, Reserved, Scope, Open, Close,
        Bind, Delimit, Definition, Let, Ltac, Fixpoint, CoFixpoint, Add,
        Morphism, Relation, Implicit, Arguments, Unset, Contextual,
        Strict, Prenex, Implicits, Inductive, CoInductive, Record,
        Structure, Canonical, Coercion, Context, Class, Global, Instance,
        Program, Infix, Theorem, Lemma, Corollary, Proposition, Fact,
        Remark, Example, Proof, Goal, Save, Qed, Defined, Hint, Resolve,
        Rewrite, View, Search, Show, Print, Printing, All, Eval, Check,
        Projections, inside, outside, Def},
    morekeywords=[2]{forall, exists, exists2, fun, fix, cofix, struct,
        match, with, end, as, in, return, let, if, is, then, else, for, of,
        nosimpl, when},
    morekeywords=[3]{Type, Prop, Set, true, false, option},
    morekeywords=[4]{pose, set, move, case, elim, apply, clear, hnf,
        intro, intros, generalize, rename, pattern, after, destruct,
        induction, using, refine, inversion, injection, rewrite, congr,
        unlock, compute, ring, field, fourier, replace, fold, unfold,
        change, cutrewrite, simpl, have, suff, wlog, suffices, without,
        loss, nat_norm, assert, cut, trivial, revert, bool_congr, nat_congr,
        symmetry, transitivity, auto, split, left, right, autorewrite},
    morekeywords=[5]{by, done, exact, reflexivity, tauto, romega, omega,
        assumption, solve, contradiction, discriminate},
    morekeywords=[6]{do, last, first, try, idtac, repeat},
    morecomment=[s]{(*}{*)},
    showstringspaces=false,
    morestring=[b]",
    morestring=[d]’,
    tabsize=3,
    extendedchars=false,
    sensitive=true,
    breaklines=false,
    basicstyle=\small,
    captionpos=b,
    columns=[l]flexible,
    identifierstyle={\ttfamily\color{black}},
    keywordstyle=[1]{\ttfamily\color{dkviolet}},
    keywordstyle=[2]{\ttfamily\color{dkgreen}},
    keywordstyle=[3]{\ttfamily\color{ltblue}},
    keywordstyle=[4]{\ttfamily\color{dkblue}},
    keywordstyle=[5]{\ttfamily\color{dkred}},
    stringstyle=\ttfamily,
    commentstyle={\ttfamily\color{dkgreen}},
    literate=
    {\\forall}{{\color{dkgreen}{$\forall\;$}}}1
    {\\exists}{{$\exists\;$}}1
    {<-}{{$\leftarrow\;$}}1
    {=>}{{$\Rightarrow\;$}}1
    {==}{{{\tt\small ==}\;}}1
    {==>}{{{\tt\small ==>}\;}}1
    {->}{{$\rightarrow\;$}}1
    {<->}{{$\leftrightarrow\;$}}1
    {<==}{{$\leq\;$}}1
    {\#}{{$^\star$}}1
    {\\o}{{$\circ\;$}}1
    {\@}{{$\cdot$}}1
    {\/\\}{{$\wedge\;$}}1
    {\\\/}{{$\vee\;$}}1
    {++}{{{\tt\small ++}}}1
    {~}{{$\sim$}}1
    {\@\@}{{$@$}}1
    {\\mapsto}{{$\mapsto\;$}}1
    {\\hline}{{\rule{\linewidth}{0.5pt}}}1
}[keywords,comments,strings]

\title{On the computational complexity of JavaScript regex matching, in Rocq}

\author{Victor Deng}{EPFL, Switzerland}{victor.deng@epfl.ch}{https://orcid.org/0000-0002-7871-0147}{}

\author{Aur\`ele Barri\`ere}{CNRS, ENS de Lyon, Inria, Université Claude Bernard Lyon 1, LIP, France}{aurele.barriere@cnrs.fr}{https://orcid.org/0000-0002-7297-2170}{}

\author{Cl\'ement Pit-Claudel}{EPFL, Switzerland}{clement.pit-claudel@epfl.ch}{https://orcid.org/0000-0002-1900-3901}{}

\authorrunning{V. Deng, A. Barri\`ere and C. Pit-Claudel} 

\Copyright{Victor Deng, Aur\`ele Barri\`ere and Cl\'ement Pit-Claudel} 



\category{} 

\relatedversion{} 

\acknowledgements{%
  The authors would like to thank Ola Svensson, Mika Göös and A. R. Balasubramanian for helpful discussions about this paper.
  This research was funded in whole or in part by the Swiss National Science Foundation (SNSF), grant number 10003649.%
}

\nolinenumbers 

\begin{document}

\maketitle

\begin{abstract}
  Despite widespread use, the complexity of the matching problem for modern regular expressions languages remains unclear.
  Previous work proved that an idealized regular expression language with backreferences and lookarounds had PSPACE-complete matching.
  We extend this work to a real-world regex language by proving that \emph{JavaScript} regex matching with expanded lower-bounded quantifiers is PSPACE-complete.  We then generalize the result: we show that PSPACE-hardness survives the removal of negative lookarounds, and that removing all lookarounds leads to an OptP-complete parsing problem.
  Our core arguments are formalized in Rocq.
\end{abstract}

\newpage

\section{Introduction}

Modern regular expressions, or regexes, enjoy widespread use, such as in over 30\% of Python and JavaScript packages~\cite{redos_impact}.
Some languages and libraries (Rust-lang/Regex, Go, RE2) adopt a regex language that permits polynomial-time matching.
Others include more advanced features, such as backreferences and atomic groups.
For the different combinations of features used in these languages, the computational complexity of the regex-matching problem is still unknown: while results exist for idealized regex languages, no work has characterized the complexity of matching real-world regexes.
To make things worse, real-world regex languages often disagree on the semantics of common features~\cite{regex_linguafranca}.

Two facts are known with respect to real-world regex matching complexity.
First, backreferences make the problem NP-hard (by reduction from 3-SAT).
While traditional regular expressions extended with backreferences alone are also NP-complete~\cite{algopatternsstrings_aho}, online sources often incorrectly generalize this NP-completeness claim for languages that contain more features~\cite{perl-3sat}.
Second, JavaScript regex matching without backreferences nor bounded quantifiers\footnote{Bounded quantifiers are quantifiers that specify their minimum and maximum number of iterations.} is in P, even with lookarounds~\cite{linearjs2024}.

More definitive results are also known for \emph{theoretical} regex languages, without being directly applicable to real-world ones.
For instance, an extension of traditional regular expressions with backreferences and lookarounds was recently shown to be PSPACE-complete~\cite{uezato_rewblk}.
This result is proved for a non-backtracking semantics that allows duplicate groups and does not have atomic lookarounds, whereas JavaScript regex matching is specified by a backtracking algorithm, has constrained duplicate capture groups, and atomic lookarounds.

In this paper, we focus on JavaScript regex matching, because JavaScript regexes are popular and have a well-defined specification, in the form of a dedicated chapter in the ECMAScript specification of JavaScript~\cite{ecma_2023,ecma_2024,ecma_2025}.
We clarify the computational complexity of JavaScript regex matching and present new results, with mechanized proofs.
Our contributions are the following:
\begin{itemize}
    \item We adapt the result of Uezato~\cite{uezato_rewblk} about PSPACE-hardness of regex matching with look\-arounds and backreferences to the real-world JavaScript regexes: we prove a reduction of a variant of the PSPACE-complete problem QBF to JavaScript regex matching.
    \item We show that negative lookarounds, which are constructs that allow negating the results of regex matching, are in fact not essential to PSPACE-hardness of JavaScript regex matching: we prove a new reduction of the same variant of QBF to JavaScript regex matching \emph{without} negative lookarounds. To do so, we show that in JavaScript (and in most modern regex languages), positive lookarounds are unexpectedly powerful, and are in fact sufficient to encode a limited form of negation.
    \item We show that JavaScript regex matching with expanded lower-bounded quantifiers is in PSPACE.
    \item We show that JavaScript regex \emph{parsing} without lookarounds (and with expanded lower-bounded quantifiers) is OptP-complete. In the problem of regex \emph{parsing}, the output is how the regex matched the string.
    \item We mechanize, in the Rocq proof assistant, the core arguments of all of these results.
\end{itemize}

Our Rocq proofs about JavaScript regexes are conducted on a preexisting mechanized semantics of JavaScript regular expressions~\cite{linden_popl26}.
In Appendix~\ref{sec:warblre_connect}, we show how our proofs have also been ported to an equivalent preexisting Rocq semantics, that is a direct manual line-by-line translation of the ECMAScript regex chapter~\cite{warblre_icfp}.
Accordingly, our Rocq proofs ensure that our results are directly applicable to a real-world regex language, despite its complex backtracking semantics.
To our knowledge, this is the first time that complexity class results are proven about a real-world regex specification and that such proofs are mechanized in a proof assistant.

The rest of this paper is organized as follows.
Section \ref{sec:background} provides background about JavaScript regexes and the semantics we use.
In the next two sections, we give key intuitions and examples for our results, and refer to the corresponding appendices and the Rocq development, available at \url{https://github.com/LindenRegex/js-regex-matching-complexity}, for details.
Section \ref{sec:pspace-complete} shows that JavaScript regex matching is PSPACE-complete.
Section \ref{sec:optp-complete} shows that JavaScript regex \emph{parsing} without lookarounds is OptP-complete.
Section \ref{sec:relatedwork} discusses related work, Section \ref{sec:discussion} discusses our findings, its limitations and future work, and Section \ref{sec:conclusion} concludes the paper.

\section{Background}
\label{sec:background}

\subsection{JavaScript regular expressions}

Modern regular expressions, such as those found in JavaScript, include many features: besides the traditional characters, concatenation, disjunction, Kleene star and Kleene plus, one can find anchors, capture groups, lookarounds and backreferences, among others.
We focus here on JavaScript regular expressions, whose features (as of the 2023 edition of the ECMAScript specification~\cite{ecma_2023}) are listed in Figure~\ref{fig:jsregex_syntax}.

\begin{figure}[hbt!]
    \raggedright
    \begin{minipage}[t]{.43\textwidth}
        \begin{tabular}{@{}l r l l}
        \re{\subreg} & $::=$ & \re{\regepsilon} & Empty \\
        &$\mid$& \re{\rchar} & Character \\
        &$\mid$& \re{\disjunction{\subreg_1}{\subreg_2}} & Disjunction \\
        &$\mid$& \re{\sequence{\subreg_1}{\subreg_2}} & Sequence \\
        &$\mid$& \re{\group{\gid}{\subreg}} & Capture group \\
        &$\mid$& \re{\noncap{\subreg}} & Non-capturing group \\
        &$\mid$& \re{\backref{\gid}} $\mid$ \re{\backref{k\langle{}\gid\rangle}} & Backreference \\
        &$\mid$& \re{\regstar{\subreg}} & Star \\
        &$\mid$& \re{\neglookahead{\subreg}} & Negative lookahead \\
        \end{tabular}
    \end{minipage}
    \begin{minipage}[t]{.49\textwidth}
        \begin{tabular}{l r l l}
        &$\mid$& \re{\lookahead{\subreg}} & Positive lookahead \\
        &$\mid$& \re{\regbos} ~ $\mid$ ~ \re{\regeos} ~ $\mid$ ~ \re{\esc{b}} ~ $\mid$ ~ \re{\esc{B}} & Anchors \\
        &$\mid$& $[c_0c_1-c_2]$ & \multirow{3}{*}{Character descriptors} \\
        &$\mid$& $[\cdcomplement{}c_0c_1-c_2]$ & \\
        &$\mid$& \re{\esc{w}} ~ $\mid$ ~ \re{\esc{W}} ~ $\mid$ ~ ... & \\
        &$\mid$& \re{\lookbehind{\subreg}} ~ $\mid$ \re{\neglookbehind{\subreg}} & Lookbehinds \\
        &$\mid$& \re{\quant{\subreg}{\rmin}{\Delta}{\top}} & Greedy quantifiers \\
        &$\mid$& \re{\quant{\subreg}{\rmin}{\Delta}{\bot}} & Lazy quantifiers
        \end{tabular}
    \end{minipage}
    \caption{JavaScript regex syntax (adapted for presentation purposes). Left: subset of constructs used in our proof of PSPACE-hardness. Right: other regex constructs. Figure adapted from~\cite{linden_popl26}.}
    \label{fig:jsregex_syntax}
\end{figure}

The syntax of JavaScript regexes described in Figure \ref{fig:jsregex_syntax} is adapted for presentation purposes as in~\cite{linden_popl26}.
Notably:
\begin{itemize}
      \item We annotate all capture groups with a group index or group name instead of relying on the implicit numbering from left to right of their opening parentheses. For instance, the original JavaScript regex \texttt{(a)(b)cd} corresponds to the regex \re{\group{1}{a}\group{2}{b}cd} in our syntax.
      \item We use a different syntax for quantified regexes.
      In \re{\quant{\subreg}{\rmin}{\Delta}{\greedy}}, $\rmin$ is the minimum number of repetitions, $\Delta$ (possibly infinite) corresponds to the difference between the maximum and the minimum number of repetitions, and $p = \top$ (resp. $p = \bot$) denotes a greedy (resp. lazy) quantifier.
\end{itemize}

The JavaScript regex features of interest that are not traditional regex features are as follows.
A \emph{capture group} records the part of the input string that was last matched by a subregex.
For instance, matching the regex \re{a\group{1}{\regstar{b}}c} on string \str{abbbc} returns the entire string as well as the substring \str{bbb} corresponding to capture group 1. If a capture group has not been matched, an engine should return an undefined value: \vundef.
Capture groups can optionally be named.
A \emph{lookaround} is an assertion that checks whether a regex matches the next (lookahead) or previous (lookbehind) part of the string; as an assertion, it does not consume characters in the string during matching.
For instance, matching the regex \re{a\lookahead{bc}b} on the string \str{abc} returns the string \str{ab}: indeed, the lookahead assertion \re{\lookahead{bc}} succeeds when matched at position 1 (between the \chr{a} and the \chr{b}), since it is followed by the string \str{bc}, but the matching process then continues by matching the remaining regex \re{b} starting from position 1.
Lookarounds can be negative, in which case they succeed if and only if the underlying regex does \emph{not} match the string at the current position.
For instance, matching the regex \re{a\neglookahead{bc}b} on \str{abc} does not succeed.
Finally, a backreference matches a substring that was previously recorded by a capture group.
For instance, the regex \re{\group{1}{a\regstar{b}c}\backref{1}} matches any string of the form \str{ab...bc} repeated twice; it matches the string \str{abbcabbc} (with capture group 1 set to \str{abbc}) but \emph{not} the string \str{abbcabc}. When a backreference refers to a group that is undefined, it matches the empty string.

\paragraph*{Regex size in the presence of quantifiers} \label{subsec:regex_size} When measuring the size of a regular expression, we count lower-bounded quantifiers as if they had been expanded away: for example, we treat \re{\regplus{\subreg}} as \re{\sequence{\subreg}{\regstar{\subreg}}} (this is equivalent to restricting ourselves to regexes without lower-bounded quantifiers).  More precisely, we define the size function as follows:

\begin{figure}[hbt!]
    \centering
    \(\begin{aligned}
    |\re{\regepsilon}| = |\re{\regchar{\cd}}| &= 1 \\
    |\re{\anchor{\anc}}| = |\re{\backref{\gid}}| &= 1 \\
    |\re{\noncap{\disjunction{\subreg_1}{\subreg_2}}}| = |\re{\sequence{\subreg_1}{\subreg_2}}| &= 1 + |\re{\subreg_1}| + |\re{\subreg_2}|
    \end{aligned}\)\qquad%
    \(\begin{aligned}
    \mathcolorbox{yellow}{|\re{\quant{\subreg}{\rmin}{\Delta}{\greedy}}|} &= \mathcolorbox{yellow}{(1 + \rmin) \times (3 + |\subreg|)} \\
    |\re{\lookaround{\lk}{\subreg}}| &= 1 + |\re{\subreg}| \\
    |\re{\group{\gid}{\subreg}}| &= 2 + |\re{\subreg}|
    \end{aligned}\)
    \caption{Definition of the regex \emph{expanded} size. Nonlinearity in terms of AST size highlighted.}
    \label{fig:regex_size_expanded}
\end{figure}

\paragraph*{Semantics} JavaScript regexes follow \textit{backtracking semantics} (also called \textit{leftmost-greedy}), meaning that the result is specified by a backtracking algorithm. This backtracking algorithm tries matching the regex from each starting position from left to right, and within a starting position, it prioritizes the left branch of each disjunction, and only tries the right branch if the left branch fails to find a match.

\subsection{The Backtracking Tree Semantics for JavaScript regular expressions}
In this work, we build on the inductive big-step semantics of JavaScript regexes presented in~\cite{linden_popl26}.
Intuitively, it relates a regex and a string to a \emph{backtracking tree}, which describes the steps and choices performed by a backtracking algorithm, but without stopping at the first match as prescribed by the backtracking semantics. The branches of that tree are ordered: the leftmost branch corresponds to the path that would be explored first by the backtracking algorithm.
More precisely, the inductive statement of the semantics has the following form:
$ \istree{\actions}{\inp}{\gm}{\dir}{\treecont} $
where $\actions$ is a list of \emph{actions} left to be performed to finish matching, $\inp$ is the current input (an input string and a position within the string), $\gm$ is a \emph{group map} that stores the optional values of each capture group in the regex, $\dir$ is a matching direction (forward or backward) and $\treecont$ is a backtracking tree.
There are three types of actions: matching a regex $r$, which is simply denoted by $\areg{r}$, closing a capture group $\gid$, denoted by $\aclose{\gid}$, and checking for progress with respect to a certain input $\inp$ , denoted by $\acheck{\inp}$ (to implement the JavaScript nullable quantifier semantics, which forbids optional iterations of a quantifier to match the empty string). Initially, the list of actions contains only the input regex.

\begin{figure}[hbt!]
    \centering
    \begin{minipage}{.7\textwidth}
        \centering
        \mbox{\infer[\ruledef{tree}{Match}]{\istree{[]}{\inp}{\gm}{\dir}{\treematch}}{}}\hfill%

        \semspace
        \mbox{\infer[\ruledef{tree}{Close}]{\istree{\aclose{\gid} :: \cont}{\inp}{\gm}{\dir}{\treeclose{\gid}{\treecont}}}
            {\istree{\cont}{\inp}{\gmclose{\gm}{\gid}{\idx{\inp}}}{\dir}{\treecont}}}%

        \semspace
        \mbox{\infer[\ruledef{tree}{Disj}]{\istree{\areg{\re{(\disjunction{\subreg_1}{\subreg_2})}} :: \cont}{\inp}{\gm}{\dir}{\treechoice{\treecont_1}{\treecont_2}}}
            {\istree{\areg{\re{\subreg_1}} :: \cont}{\inp}{\gm}{\dir}{\treecont_1} \newpremise \istree{\areg{\re{\subreg_2}} :: \cont}{\inp}{\gm}{\dir}{\treecont_2}}}%

        \semspace
        \mbox{\infer[\ruledef{tree}{SeqForward}]{\istree{\areg{\re{\sequence{\subreg_1}{\subreg_2}}} :: \cont}{\inp}{\gm}{\forward}{\treecont}}
            {\istree{\areg{\re{\subreg_1}} :: \areg{\re{\subreg_2}} :: \cont}{\inp}{\gm}{\forward}{\treecont}}}

        \semspace
        \mbox{\infer[\ruledef{tree}{GreedyStar}]{\istree{\areg{\re{\regstar{\subreg}}} :: \cont}{\inp}{\gm}{\dir}{\treechoice{(\treereset{\defgroups{\subreg}}{\treeiter})}{\treecont_{skip}}}}
        {\begin{array}{c} \istree{\cont}{\inp}{\gm}{\dir}{\treecont_{skip}} \\ \istree{\areg{\re{\subreg}} :: \acheck{\inp} :: \areg{\re{\regstar{\subreg}}} :: \cont}{\inp}{\gmreset{\gm}{\defgroups{\subreg}}}{\dir}{\treeiter}\end{array}}}%
    \end{minipage}
    \begin{minipage}{.28\textwidth}
        \def\yunit{0.7}
        \begin{tikzpicture}
            \begin{pgfonlayer}{foreground}
                \node (root) at (0, 0) {Open 1};
                \node (choice) at (0, -1*\yunit) {Choice};
                \node (reada_left) at (-1, -2*\yunit) {Read a};
                \node (close1_left) at (-1, -3*\yunit) {Close 1};
                \node (mismatch) at (-1, -4*\yunit) {Mismatch};
                \node (reada_right) at (1, -2*\yunit) {Read a};
                \node (readb) at (1, -3*\yunit) {Read b};
                \node (close1_right) at (1, -4*\yunit) {Close 1};
                \node (readc) at (1, -5*\yunit) {Read c};
                \node (match) at (1, -6*\yunit) {Match};
                \draw (root) -- (choice);
                \draw (choice) -- (reada_left);
                \draw (reada_left) -- (close1_left);
                \draw (close1_left) -- (mismatch);
                \draw (choice) -- (reada_right);
                \draw (reada_right) -- (readb);
                \draw (readb) -- (close1_right);
                \draw (close1_right) -- (readc);
                \draw (readc) -- (match);
            \end{pgfonlayer}
        \end{tikzpicture}
    \end{minipage}
    \caption{Left: examples of inductive rules of the semantics defined in~\cite{linden_popl26}. Right: backtracking tree of regex $(_1 \mathrm{a|ab})\mathrm{c}$ on string \str{abc}.}
    \label{fig:linden_semantics}
\end{figure}

Figure~\ref{fig:linden_semantics} (left) presents a few inductive rules of the semantics, where $\gmclosenoarg$ and $\gmresetnoarg$ manipulate the group map.
For instance, rule \ruleref{tree}{Match} states that the backtracking tree corresponding to no action at all is simply a \treematch~node.
Rule \ruleref{tree}{Close} states that the backtracking tree for the list of actions $\aclose{\gid} :: \cont$ is a $\treeclose{\gid}{\treecont}$ node, whose subtree $\treecont$ is the backtracking tree corresponding to some list of actions $\actions$.
Finally, the simplified rule \ruleref{tree}{GreedyStar} creates two branches in the tree, under a $\treechoice{}{}$ node.
The left branch $t_{iter}$ corresponds to iterating the greedy star, while the right one $t_{skip}$ simply ignores the star and matches the remaining actions $\actions$.
Inside the actions of the left branch, an action $\acheck{\inp}$ is inserted to check that in that iteration, progress is made with respect to the original input $\inp$.
We refer the reader to~\cite{linden_popl26} for the full, precise semantics.

For instance, the backtracking tree of the regex \re{\group{1}{\disjunction{a}{ab}}c} and the string \str{abc} is shown on Figure~\ref{fig:linden_semantics} (right), where the $\treechoice{}{}$ node is created by the disjunction.

This backtracking tree semantics allows to recover the result of JavaScript regex matching, by considering the branches of the tree that end on a \treematch~node.
For each of these branches, one can compute a match result, called a \textit{leaf}~\cite{linden_popl26}: a pair of an input and a group map.
For instance, in Figure~\ref{fig:linden_semantics} (right), there is only one such branch.
In its associated leaf, the input is at position 3 (there are 3 \ctorone{Read}{} nodes on that branch), and the group map contains a group 1, that has been opened at position 0, and closed at position 2.
The \textit{first leaf} of a backtracking tree is the leaf associated with the leftmost branch that ends on a \treematch~node: it is the match of highest priority and corresponds to the result of JavaScript regex matching.

\section{JavaScript regex matching is PSPACE-complete}
\label{sec:pspace-complete}

In this section, we show that the decision version of JavaScript regex matching is PSPACE-complete.
We call this decision problem \jsregmatch:

\textbf{Input:} A regex $r$ and a string $s$.
\qquad
\textbf{Problem:} Does $r$ have an anchored match on $s$?

An anchored match of $r$ on $s$ is a match of $r$ on $s$ that starts at the beginning of $s$.
In JavaScript, one can request an anchored match by setting the \emph{sticky} (\chr{y}) flag to true.
We prove that \jsregmatch is PSPACE-hard (Section~\ref{subsec:pspace-hard}), and that it is in PSPACE (Section~\ref{subsec:pspace-membership}).
We discuss how to extend our result to unanchored matching in Appendix \ref{app:unanchored_anchored}.

\subsection{JavaScript regex matching is PSPACE-hard}
\label{subsec:pspace-hard}

In this section, we adapt the result of Uezato~\cite{uezato_rewblk} about PSPACE-hardness of regex matching with lookarounds and backreferences to JavaScript regex matching.
As in Uezato's work, we prove PSPACE-hardness of \jsregmatch by reducing a variant of the PSPACE-complete problem QBF (Quantified Boolean Formula) to it.

\paragraph*{QBF and variants}
QBF is the canonical PSPACE-complete problem~\cite{computers_and_intractability}.
The input to this problem is a quantified boolean formula, which is a formula that comprises variables $x, y$, logical connectors $\neg$, $\land$ and $\lor$ and quantifiers $\forall$ and $\exists$, and where all variables are bound by a quantifier.
Its output is whether the quantified boolean formula is true or false.
For instance, $\forall x. \exists y. x \lor y$ is a true QBF (take for instance $y = \top$), while $\forall x. \exists y. x$ is a false QBF (take $x = \bot$).

It is well-known that the PCNF-QBF variant is also PSPACE-complete.
In this variant, the quantified boolean formula is in prenex form: all quantifiers are at the beginning of the formula, and its underlying propositional formula is in conjunctive normal form.

Finally, as in Uezato's work~\cite{uezato_rewblk}, we consider another variant \qbfprime of the problem PCNF-QBF, which is also PSPACE-complete. In this variant, the only quantifiers are $\exists$ and $\neg\exists$, the underlying propositional formula is in conjunctive normal form (like in PCNF-QBF) and, unlike in PCNF-QBF, the underlying propositional formula can be negated.

\qbfprime is in PSPACE because QBF is, and it is PSPACE-hard because one can easily reduce PCNF-QBF to it in linear time, by replacing $\forall x.$ quantifiers by $\neg\exists x. \neg$ and eliminating double negations.
We prove this reduction from PCNF-QBF to \qbfprime correct in Rocq:
\begin{theorem}
    Let $q$ be a PCNF-QBF formula and $q'$ its translation into a \qbfprime formula.
    Then $q'$ and $q$ have the same size, and $q$ is true if and only if $q'$ is true.
\end{theorem}

\paragraph*{Our reduction from \qbfprime to \jsregmatch}

In this section, we explain our reduction from \qbfprime to \jsregmatch using an example.
The example is as follows: $q$ is a \qbfprime formula, and we build the regex $r$ and the string $s$.
$$
    q = \exists x_1. \neg\exists x_2. (x_1 \lor \neg x_2) \land x_2 \quad
    r = \re{\noncap{\disjunction{\group{1}{a}}{a}}{;}\neglookahead{\noncap{\disjunction{\group{2}{a}}{a}}{;}\noncap{\disjunction{\backref{1}}{\backref{2}a}}{;}\backref{2}{;}}} \quad
    s = \str{a;a;a;a;z}
$$
With our reduction, the backtracking regex matching process simulates a naive \qbfprime checking algorithm.
As the regex is being matched, the group map models a valuation $v : \{x_1, x_2\} \to \{\top, \bot\}$ being tested, by mapping each variable $x_i, i \in \{1, 2\}$ to either the string \str{a} (for $\top$) or \vundef, that a backreference will treat as the empty string \str{} (for $\bot$).
The \textit{assignment fragments} of the regex, $\re{\noncap{\disjunction{\group{\mathit{i}}{a}}{a}}{;}}$ (for $i = 1, 2$) force the backtracking semantics to consider all the assignments of the capture groups 1 and 2 to \str{a} or \str{}, effectively testing all possible valuations.
The \textit{core fragment}, $\re{\noncap{\disjunction{\backref{1}}{\backref{2}a}}{;}\backref{2}{;}}$, when matched against the substring \str{a;a;}, checks whether the current valuation makes the CNF $(x_1 \lor \neg x_2) \land x_2$ true.
To do so, each subfragment $\re{\noncap{\disjunction{\backref{1}}{\backref{2}a}}{;}}$ and $\re{\backref{2}{;}}$ must match the substring \str{a;}. In these subfragments, each disjunct corresponds to one literal, with a backreference $\re{\backref{\mathit{i}}}$ matching the string \str{a} if and only if variable $x_i$ is true and a backreference $\re{\backref{\mathit{i}}a}$ matching the string \str{a} if and only if variable $x_i$ is false.
Finally, the negative lookahead encodes the negated existential quantifier.
After the first assignment fragment, the regex $\re{\neglookahead{\noncap{\disjunction{\group{2}{a}}{a}}{;}\noncap{\disjunction{\backref{1}}{\backref{2}a}}{;}\backref{2}{;}}}$ matches \str{a;a;a;} if and only if the current valuation makes the partial \qbfprime $\neg\exists x_2. (x_1 \lor \neg x_2) \land x_2$ true.

The final \chr{z} in the string is not used in the reduction, it is only here to ease proof sharing with the next subsection.
It has no incidence on matching, since a regex does not need to match the entire string to have an anchored match on that string.

We prove in Rocq the following theorem, which proves the correctness of the reduction from \qbfprime to JavaScript regex matching:
\begin{theorem}
Let $q$ be a \qbfprime formula, and $(r, s)$ the regex and the string that result from the translation of $q$ into an instance of regex matching. Then $r$ has an anchored match on $s$ if and only if $q$ is true.
\end{theorem}

We refer the reader to Appendix \ref{appsubsec:formal_constr_pspace} for the formal construction.

\subsection{\jsregmatch is PSPACE-hard, even without negative lookarounds}
\label{subsec:pspace-hard-noneglk}

We found that negative lookarounds, while the most natural way of modeling negation, were not essential to do so.
Instead, we leverage the \emph{atomicity} of lookarounds in JavaScript (and other regex languages) to provide a reduction from \qbfprime to JavaScript regex matching without negative lookarounds.
This atomicity means that lookarounds act as a backtracking frontier: when matching a regex $\re{\lookahead{\mathit{A}}\mathit{B}}$, if the backtracking semantics finds some match for the regex $A$ but then fails to find a match for the regex $B$, it will \emph{not} backtrack inside the lookahead to try another match for $A$.
For instance, there is no match of $\re{\lookahead{\group{1}{\disjunction{a}{ab}}}\backref{1}c}$ on the string \str{abc}: even though taking the second branch of the lookahead would set group 1 to \str{ab} and allow $\re{\backref{1}c}$ to match, this branch is never visited.

We combine this atomicity with the fact that with backtracking semantics, when matching a regex $A|B$, if we know that $B$ matched, then we also know that $A$ did \emph{not} match.
As a consequence, one can test whether a regex $A$ does \emph{not} match a string $s$ by matching a regex of the form $\re{\lookahead{\disjunction{\mathit{A}}{\mathit{B}}}\mathit{C}}$, where the regex $C$ checks whether the regex $B$ was matched in the lookahead.
If $B$ was matched, it means that $A$ did \emph{not} match, since otherwise the backtracking semantics would not consider matching $B$.
Conversely, if $B$ did not match in the lookahead, it means that $A$ did match.

We explain our new negative-lookbehind-free reduction using the following example:
\begin{align*}
    q &= \exists x_1. \neg\exists x_2. x_1 \lor x_2 &
    r &= \re{\noncap{\disjunction{\group{1}{a}}{a}}{;}\lookahead{\disjunction{\noncap{\disjunction{\group{2}{a}}{a}}{;}\noncap{\disjunction{\backref{1}}{\backref{2}}}{;}}{\regstar{\noncap{a;}}\group{z}{z}}}\regstar{\noncap{a;}}\backref{k\langle z\rangle}\anchor{\regeos}} \\
    & & s &= \str{a;a;a;z}
\end{align*}

The role of the assignment fragments $\re{\noncap{\disjunction{\group{\mathit{i}}{a}}{a}}{;}}$ remains the same as in Section~\ref{subsec:pspace-hard}: the backtracking semantics has to consider all valuations of $x_i$.
The core fragment $\re{\noncap{\disjunction{\backref{1}}{\backref{2}}}{;}}$ keeps its purpose as well: it matches the string \str{a;} if and only if the current valuation makes the CNF $x_1 \lor x_2$ true.
What changes is the way negation is modeled.
Consider a group map $\gm$ that has group $1$ set to either \str{a} or \vundef and all other groups undefined, and the corresponding agreeing valuation $v$.
Assume that the first assignment fragment $\re{\noncap{\disjunction{\group{1}{a}}{a}}{;}}$ has been matched against the string \str{a;}, and we have a current valuation $v$ for the variable $x_1$.
\begin{itemize}
    \item Assume that $v$ makes the partial \qbfprime $\exists x_2. x_1 \lor x_2$ true.
    Then the sub-regex $\re{\noncap{\disjunction{\group{2}{a}}{a}}{;}\noncap{\disjunction{\backref{1}}{\backref{2}}}{;}}$ matches the string \str{a;a;} and the positive lookahead $\re{\lookahead{\disjunction{\noncap{\disjunction{\group{2}{a}}{a}}{;}\noncap{\disjunction{\backref{1}}{\backref{2}}}{;}}{\regstar{\noncap{a;}}\group{z}{z}}}}$ matches (the empty string) leaving group z undefined; consequently, the fragment $\re{\regstar{\noncap{a;}}\backref{k\langle z\rangle}\anchor{\regeos}}$ does not have an anchored match in the string \str{a;a;z}.
    Because of the atomicity of lookaheads, the backtracking semantics will \emph{not} backtrack inside the lookahead and match the other branch $\re{\regstar{\noncap{a;}}\group{z}{z}}$ against the string \str{a;a;z}.
    The matching process therefore fails to find a match.
    \item Assume conversely that $v$ makes the partial \qbfprime $\exists x_2. x_1 \lor x_2$ false.
    Then the sub-regex $\re{\noncap{\disjunction{\group{2}{a}}{a}}{;}\noncap{\disjunction{\backref{1}}{\backref{2}}}{;}}$ does not have a match on the string \str{a;a;z}, and the positive lookahead $\re{\lookahead{\disjunction{\noncap{\disjunction{\group{2}{a}}{a}}{;}\noncap{\disjunction{\backref{1}}{\backref{2}}}{;}}{\regstar{\noncap{a;}}\group{z}{z}}}}$ matches with group z set to the string \str{z}; consequently, the fragment $\re{\regstar{\noncap{a;}}\backref{k\langle z\rangle}\anchor{\regeos}}$ matches the string \str{a;a;z}.
    The matching process therefore succeeds in finding a match.
\end{itemize}
Finally, the matching process finds a match if and only if $v$ makes the partial \qbfprime $\exists x_2. x_1 \lor x_2$ false, effectively negating the output of matching for the sub-regex $\re{\noncap{\disjunction{\group{2}{a}}{a}}{;}\noncap{\disjunction{\backref{1}}{\backref{2}}}{;}}$.

We prove in Rocq the following theorem, that proves the correctness of the reduction from \qbfprime to JavaScript regex matching:
\begin{theorem}
    Let $q$ be a \qbfprime formula. Let $(r, s)$ be the regex and string that result from the translation of $q$ into a regex without negative lookarounds and a string. Then $r$ has an anchored match on $s$ if and only if $q$ is true.
\end{theorem}

We refer the reader to Appendix \ref{appsubsec:formal_constr_pspace_noneglk} for the full formal construction.

\subsection{JavaScript regex matching is in PSPACE}
\label{subsec:pspace-membership}

To prove PSPACE membership of \jsregmatch, we present a backtracking algorithm.
This algorithm relies on a partial function, \computeresult($\actions, \inp, \gm, \dir, \fuel$), that computes the first leaf of the tree $\treecont$ such that $\istree{\actions}{\inp}{\gm}{\dir}{\treecont}$.
The function is presented in its entirety in Appendix~\ref{app:pspace_membership_pseudocode}.
The function is partial: the $\fuel$ argument strictly decreases at each recursive call, and the function only returns a result when the recursion depth does not exceed the initial fuel.
This function is directly inspired by the \emph{functional semantics} of Barrière et al.~\cite{linden_popl26}, but returning the first leaf instead of the full tree.

An example of the recursive definition is shown below:
\begin{lstlisting}[language=Coq,backgroundcolor=\color{white}]
compute_result ($(\subreg_1|\subreg_2) :: \cont$) $\inp$ $\gm$ $\dir$ ($\fuel + 1$) =
    match compute_result ($\subreg_1 :: \cont$) $\inp$ $\gm$ $\dir$ $\fuel$ with
    | Out_of_fuel => Out_of_fuel
    | Success $\lf$ => Success $\lf$
    | NoMatch => compute_result ($\subreg_2 :: \cont$) $\inp$ $\gm$ $\dir$ $\fuel$
    end
\end{lstlisting}

To show that we can solve \jsregmatch using polynomial space, we find and prove polynomial bounds (in the sizes of the regex and the input string) both for the size of the arguments of \computeresult, and also for an amount of initial fuel that is sufficient for the partial function to always return a result.
We present these bounds in Appendix \ref{app:pspace_membership_pseudocode}.
These bounds are sufficient to show that the algorithm uses polynomial space: the function execution only stores the current arguments and a call stack. The arguments require polynomial space, and the call stack contains at most a polynomial number of stack frames, each using polynomial space to store function arguments.

The key difficulty resides in polynomially bounding the recursion depth of \computeresult.
Barrière et al. \cite{linden_popl26} previously presented a recursion depth bound for their functional semantics, but their bound is exponential in the size of the AST of the regex.
Intuitively, they defined a fuel function $\fuelf{r}{i}$, over-approximating the depth of the backtracking tree of regex $r$ with input string $i$.
For the greedy star $\re{\regstar{\subreg}}$, they relied on the following observation: since each iteration of the star either reads a character or ends up being discarded by a failed progress check, then there can be at most $\strsize{i}+1$ iterations of $r$.
They consequently defined $\fuelf{\re{\regstar{\subreg}}}{i}$ as $(\fuelf{r}{i}+2)\times(\strsize{i}+1)$. This formula can be generalized to quantifiers with forced iterations (\re{\quant{\subreg}{\rmin}{\Delta}{b}} when $\rmin>0$) by adding $\rmin$ to the number of iterations.
The issue with this multiplication happens when stars are nested: $\fuelf{\re{\regstar{\regstar{\subreg}}}}{i}$ is quadratic, and the fuel of $n$ nested stars is $\Omega(\strsize{i}^n)$.

To avoid this exponential complexity, we define a new, more complex fuel formula that leverages the following new observation: every valid non-forced iteration of a nested quantifier consumes a character, and therefore decreases the number of possible non-forced iterations of outer quantifiers.
In other words, when matching $\re{\regstar{\regstar{\subreg}}}$ on input $i$, the subregex $r$ can only be matched $\strsize{s}+1$ times.
The formulae for the bounds and a proof of adequacy of the bounds can be found in Appendix \ref{app:pspace_membership_metrics}.
We prove in Rocq the following theorems:
\begin{theorem}
    Let $r$ be a regex and $s$ be an input string.
    Let $\fuelf{r}{s} = 1 + (1 + |s|)(|r| + |r|^2)$, where $|r|$ is the expanded size of $r$ (Figure~\ref{fig:regex_size_expanded}).
    Then running the backtracking algorithm with $r$, $s$ and $\fuelf{r}{s}$ succeeds, yielding a boolean result that matches the result of the instance $(r, s)$ of \jsregmatch.
\end{theorem}
\begin{theorem}
    Let $r$ be a regex.
    Let $\actions, \inp, \dir$ be a list of actions, input and matching direction that can result from matching $r$ on a string.
    Consider matching the list of actions $\actions$ on the input $\inp$ in direction $\dir$ using the backtracking algorithm.
    Let $\framesize = (1 + |\inp|)|\actions|$ be an upper bound on the memory usage of a stack frame of the backtracking algorithm, and $P = (1+|\inp|)(|r| + |r|^2)$ be a polynomial in the regex and input sizes.
    Then $\fuelf{r}{i} \times \framesize \le P(P+1)$.
\end{theorem}

\section{JavaScript regex parsing without lookarounds is OptP-complete}
\label{sec:optp-complete}

In this section, we prove that JavaScript regex \emph{parsing} without lookarounds is OptP-complete.

OptP~\cite{krentel_optp} is the following class of optimization problems:
\begin{definition}
    Let $\Sigma$ be an alphabet with $\{0, 1\} \subseteq \Sigma$.
    A function $f : \Sigma^* \to \mathbb{N}$ is in OptP if and only if there exists a nondeterministic Turing machine $M$ that runs in polynomial time, on each branch writes a binary number and accepts, and such that for all $x \in \Sigma^*$, $f(x)$ is the maximum (for a maximization problem) number written by a branch of $M$ on input $x$.
\end{definition}

JavaScript regex \emph{parsing} can be defined as the following problem:

\textbf{Input:} A regex $r$ and a string $s$.

\textbf{Output:} a binary number with $\nbits(r, s)+1$ bits, where $\nbits(r, s) = (1 + |s|)\times\regsize{\subreg} + 1$.
If there is an anchored match of $r$ on $s$, then the list of binary backtracking choices made while matching the regex against the string, preceded by a 1, padded to the right using ones; otherwise, 0 followed by $\nbits(r, s)$ ones.
A 1 means taking the higher-priority branch, a 0 means taking the lower-priority branch.
Most significant bit first.

This can therefore be viewed as a function from the set of pairs of a regex and a string to the set of natural numbers.
The output binary number encodes how the regex matched the string, and in particular which branch of each disjunction was taken and the iterating decisions made for each quantifier.
Therefore, the range matched by each subregex can be recovered from this number, the regex and the string.

\subsection{JavaScript regex parsing without lookarounds is OptP-hard}
\label{subsec:optp-hard}

In this section, we prove that JavaScript regex parsing without lookarounds is OptP-hard by metrically reducing a variant of the OptP-complete problem \lexsat (called \textsc{Maximum Satisfying Assignment} in \cite{krentel_optp}) to JavaScript regex parsing.
\lexsat is the following optimization problem:

\textbf{Input:} A propositional formula $\varphi$.
It must use exactly all variables from $x_1$ to $x_n$ for some $n \in \mathbb{N}$.

\textbf{Output:} The lexicographically maximum assignment of variables from $x_1$ to $x_n$ that makes $\varphi$ true, or $\bot$ if there is no such assignment.

Without loss of generality, we force the propositional formula to not contain \emph{fictive} variables, i.e. variables that would be present in assignments but absent from the formula \cite{lexsat}.
Indeed, in the presence of such variables, one can use an increasing renaming $\rho$ of the non-fictive variables to some contiguous range $\{1, ..., n'\}$, solve the new instance of \lexsat obtained, then rename the variables of the result using $\rho^{-1}$ and (if desired) set the fictive variables to 1 in the result.

In the variant CNF \lexsat of \lexsat that we consider, the formula is in conjunctive normal form.
This variant is also OptP-hard \cite[Thm. 2.1]{lexsat}.

A metric reduction is defined by Krentel~\cite{krentel_optp} as follows:
\begin{definition}
    Let $f, g: \Sigma^* \to \mathbb{N}$. A metric reduction from $f$ to $g$ is a pair of polynomial-time computable functions $T_1 : \Sigma^* \to \Sigma^*$ and $T_2 : \Sigma^* \times \mathbb{N} \to \mathbb{N}$ such that for all $x \in \Sigma^*$, $f(x) = T_2(x, g(T_1(x)))$.
\end{definition}
Essentially, to reduce a function $f$ to a function $g$, one must be able for each instance $x \in \Sigma^*$ to compute in polynomial time an instance $T_1(x)$ such that one can retrieve in polynomial time $f(x)$ from $g(T_1(x))$ and $x$.

The reduction from CNF \lexsat to JavaScript regex parsing is essentially the same as in Section \ref{subsec:pspace-hard}, in the special case where all the quantifiers are $\exists$ (no $\neg \exists$ quantifier).
Below is an example:
$$
    \varphi = (x_1 \lor \neg x_2) \land x_2 \quad
    r = \re{\noncap{\disjunction{\group{1}{a}}{a}}{;}\noncap{\disjunction{\group{2}{a}}{a}}{;}\noncap{\disjunction{\backref{1}}{\backref{2}a}}{;}\backref{2}{;}} \quad
    s = \str{a;a;a;a;z}
$$

JavaScript's backtracking semantics makes the backtracking engine consider the assignments of capture groups $1$ and $2$ to \str{a} or \vundef in reverse lexicographic order.
Indeed, the left branch of each disjunction $\re{\noncap{\disjunction{\group{\mathit{i}}{a}}{a}}{;}}$ is tried first, and that branch corresponds to assigning variable $x_i$ to the string \str{a}.
After the groups have been set up, as shown in Section \ref{subsec:pspace-hard}, the core of the regex $\re{\noncap{\disjunction{\backref{1}}{\backref{2}a}}{;}\backref{2}{;}}$ has an anchored match on the string \str{a;a;z} if and only if the agreeing valuation makes the propositional formula $\varphi$ true.
Therefore, the matching process stops at the lexicographically maximum assignment that makes $\varphi$ true and returns the corresponding backtracking choices, from which the values of the capture groups and hence the corresponding agreeing valuation can be recovered in polynomial time.

We prove in Rocq the following theorem:
\begin{theorem}
    Let $\varphi$ be a propositional formula in CNF without fictive variables. Let $(r, s)$ be the regex and string that result from translating $\varphi$ into an instance of regex matching. Then:
    \begin{itemize}
        \item $|r| \le 12|\varphi|$ and $|s| \le 2|\varphi|$, where $|\varphi| = 1 + m + \ell$, $m$ is the number of clauses in $\varphi$ and $\ell$ is the number of literal occurrences in $\varphi$,
        \item $r$ has neither lookarounds nor lower-bounded quantifiers,
        \item if there is no match of $r$ on $s$, then the \lexsat instance $\varphi$ has no solution,
        \item otherwise, the capture groups of the highest priority match, which can be recovered from the result of the instance $(r, s)$ of regex parsing, encode the solution to the original \lexsat instance.
    \end{itemize}
\end{theorem}

The formal construction can be found in Appendix \ref{app:formal_constr_optp}.

\subsection{JavaScript regex parsing without lookarounds is in OptP}
\label{subsec:optp-membership}

In this section, we show that JavaScript regex parsing without lookarounds is in OptP.
We show this by directly using the definition of OptP.
We modify the algorithm for PSPACE membership shown in Section~\ref{subsec:pspace-membership} and make it take a list of choices with a fixed polynomial length (depending on the regex and the input string) which is large enough so that the algorithm never runs out of choices to make.
This algorithm follows the list of choices to make its branching decisions, returns the list of choices preceded by a 1 if it leads to a match and returns the list of choices preceded by a 0 if it leads to a mismatch.
One bit is consumed for each recursive call, even if there is no choice to make at that recursive call.

The full algorithm pseudocode can be found in Appendix \ref{app:optp_pseudocode}.

This algorithm models a nondeterministic Turing machine that guesses the backtracking choices and returns the list of choices preceded by a 1 if they lead to a match or by a 0 if they lead to a mismatch.
We show that the algorithm takes a polynomial amount of time to run, and that the lexicographically highest list of choices that is returned by the algorithm over all possible lists of choices is indeed the desired list of choices.

More formally, we prove in Rocq the following theorem:
\begin{theorem}
    Let $r$ be a regex without lookarounds and $s$ an input string. Let $n = \nbits(r, s) = (1+|s|)\times\regsize{r} + 1$.
    Then there exists a maximal result $b$ of the above algorithm when run with lists of choices of length $n$.
    This result $b$ has length $n+1$ and corresponds to the result of matching $r$ on $s$.
\end{theorem}

\section{Related work}
\label{sec:relatedwork}

A body of work is interested in the complexity of modern regex matching.
Uezato \cite{uezato_rewblk} shows that regex matching with lookarounds and backreferences is PSPACE-complete for a natural, non-backtracking semantics.
We take inspiration from their work in this paper, extending it to a real-world backtracking semantics and providing a mechanized proof of the core arguments.
More recently, Nogami and Terauchi \cite{nogami2026hardness} show that under the Orthogonal Vectors Conjecture, regex matching with backreferences, intersection or complement cannot be solved in time $\mathrm{O}(|s|^{2-\varepsilon} \operatorname{poly}(|r|))$ for any $\varepsilon > 0$, where $|s|$ is the size of the string and $|r|$ is the size of the AST of the regular expression.
The authors did not consider lookarounds, because it was shown \cite{linearjs2024} that regex matching with lookarounds (but without backreferences nor bounded quantifiers) could be done in $\mathrm{O}(|r| \times |s|)$.
However, our results and that of Uezato show that surprisingly, while regex matching with backreferences is NP-complete (in its decision version), regex matching with lookarounds and backreferences is PSPACE-complete.
This contrasts with the lack of increase in computational complexity of the regex matching problem when only adding lookarounds without backreferences.
Finally, the construction that we formalize for OptP-hardness is directly inspired by \cite{perl-3sat}.

The complexity of other regex-related problems has also been studied.
First, the problem of succinctness of regular expressions is concerned with determining the size of the smallest regular expression (traditional or not) that can represent some language.
For instance, Ehrenfeucht and Zeiger~\cite{regex_complexity_measures} consider (among others) the language defined by the set of all paths between two nodes on the \emph{arc-labelled} complete graph on $n$ nodes, and show that any traditional regular expresssion representing this language has at least $2^{n-1}$ alphabetical symbols (among other results).
More recently, Freydenberger~\cite{ere_succinctness_decidability} shows that regular expressions with backreferences cannot be minimized effectively with respect to length or number of variables.
Gelade and Neven~\cite{succinctness_complement_intersection} show that one cannot avoid an exponential or double exponential increase in size when finding a traditional regular expression corresponding to the complement of another one, or to the intersection of a fixed or arbitrary number of other traditional regular expressions.
Gelade~\cite{succinctness_interleaving_counting} shows that one cannot avoid a double exponential increase in size when translating a regular expresssion with interleaving into a standard one.

Second, the problem of regex equivalence is concerned with determining whether two regexes represent the same language.
Stockmeyer and Meyer~\cite{stockmeyermeyer_equivalence} show that traditional regular expression equivalence is PSPACE-complete.
Together with our results, this means that one can encode an instance of traditional regex \emph{equivalence} into an instance of modern regex \emph{matching}.
Slightly more recently, Stearns and Hunt III~\cite{equivalence_unambiguous} show that traditional regular expression equivalence when restricted to \emph{unambiguous} regular expressions is decidable in polynomial time, and generalized this result to regular expressions whose degree of ambiguity is bounded by a fixed quantity.

In this paper, we relied on the mechanized semantics of JavaScript regexes presented by Barri\`ere et al.~\cite{linden_popl26}.
Other mechanized semantics for modern regexes exist: for instance, Zhuchko et al.~\cite{lean_lookarounds} formalized in Lean a model of regular expressions with lookarounds, intersection and complement, Perry~\cite{perry_thesis} formalized in Lean two disambiguation policies for regular expressions with capture groups, and Chattopadhyay et al.~\cite{rocq_lookarounds} formalized in Rocq a model of regular expressions with lookarounds.
However, these models are either incomplete (missing some features) or not explicitly related to a real-world regex language.
This is important to note, as there are semantic differences between different regex languages that make many regexes non-portable between languages \cite{regex_linguafranca}.
The semantics that we use in this work, on the other hand, has been proved equivalent to a direct and complete transcription of the ECMAScript 2023 semantics of JavaScript regular expressions~\cite{warblre_icfp} (see Appendix~\ref{sec:warblre_connect}).

Finally, complexity results can be established without reasoning with Turing machines, as we do in this paper.
Forster et al.~\cite{forster_wcbv} show that the weak call-by-value $\lambda$-calculus is a reasonable model of computation for both time and space, so that complexity results can be derived by studying the natural time (number of $\beta$-reduction steps) and space (size of largest term in computation) measures of a computation.
G\"aher and Kunze~\cite{cooklevin_coq} leverage this observation to mechanize the Cook-Levin theorem in Rocq.

\section{Limitations, future work and discussion}
\label{sec:discussion}

\paragraph*{Regex size} The membership results of Sections~\ref{subsec:pspace-membership} and \ref{subsec:optp-membership} rely on a definition of the size of a regex that expands away lower-bounded quantifiers.
This choice matches usual practice in linear-time regex matchers, which expand both lower and upper-bounded quantifiers away and measure complexity relative to this expanded size (of course, this expansion can lead to an exponential growth of the regex compared to its textual representation, and indeed backtracking matchers recurse once per regex iteration, leading to worst-case exponential stack usage).  In our case, because optional iterations must consume a character, we only need to expand lower-bounded quantifiers, not upper-bounded ones.
Bounded repetitions are an active area of research~\cite{turonova_counting,regex_bva,automata_bounded}.
We leave the complexity of matching unexpanded counted repetitions to future work.

\paragraph*{PSPACE-hardness without negation} The result of Section~\ref{subsec:pspace-hard-noneglk}, the PSPACE-hardness of JavaScript regex matching without negative lookarounds, still relies on the ability to encode some form of negation of regex matching, thanks to the atomicity of lookarounds.
To our knowledge, it is still unknown whether regex matching with non-atomic positive lookarounds and backreferences is still PSPACE-hard.
Studying the computational complexity of such regexes is left to future work.

\paragraph*{Other regex languages} Our complexity results apply to JavaScript regexes as specified in the 2023 edition of the ECMAScript specification~\cite{ecma_2023}, because the preexisting semantics we build upon~\cite{linden_popl26} formalizes that version.
Therefore, two questions naturally arise: whether our results still apply to current JavaScript regular expressions (e.g. as specified in the 2025 edition of the ECMAScript specification~\cite{ecma_2025}) and whether our results extend to other regex languages such as Perl, Java, PCRE2, etc.
For the first question, we believe that our results still apply to 2025 JavaScript regexes, as the new features neither change the semantics of the ones we use, nor seemingly increase matching complexity.
To be confident in this result, one would need to update the backtracking tree semantics, as well as the Warblre semantics~\cite{warblre_icfp} (see Appendix~\ref{sec:warblre_connect}) to ECMAScript 2025, and finally adapt our proofs to the new semantics. This is left to future work.

For the second question, there is a wide range of regex languages, each supporting different features with different semantics.
The hardness results presented in this work can be expected to apply to any real-world regex language that includes both backreferences and lookarounds.
In regex languages like the ones of PCRE2, Python, Java and .NET, backreferences fail when referencing a group that has not been defined (instead of matching the empty string, like we rely on in our reductions)~\cite{backref_reexamined}.
For these, we suggest changing the assignment fragments in our reductions, from $\re{\noncap{\disjunction{\group{\mathit{i}}{a}}{a}}{;}}$ to $\re{\noncap{\group{i}{\disjunction{a}{}}\regqm{a}}{;}}$.
As far as we know, all real-world languages that include lookarounds have atomic lookarounds, that can encode negation like in Section~\ref{subsec:pspace-hard-noneglk}. PCRE2 contains both atomic and non-atomic lookarounds, with dedicated syntax.
The membership problems would however require more care, as this work does not consider features such as \textit{balancing groups}, \textit{control verbs}, or \textit{recursion}, that can be found in some languages.
Also, some languages like POSIX regexes follow \textit{longest-match} semantics instead of backtracking semantics, changing the optimization version of the matching problem.
Unfortunately, these other real-world languages do not benefit from an existing complete mechanization in a proof assistant, and porting our work to other languages is left as future work.

\paragraph*{Relation between exec and parsing}
In Section~\ref{sec:optp-complete}, we considered the problem of regex parsing, where, given a regex and a string, the output describes how the regex matches the string.
This contrasts with the problem solved by the commonly-used JavaScript regex function \texttt{exec}, that only returns the capture ranges last matched by each capture group in the regex; we call this problem \jsregexec.
It turns out that this problem is not clearly in OptP.
Indeed, there is no ordering on the capture groups alone that can distinguish higher-priority from lower-priority matches.
Consider the following example:
$$ \subreg_1 = \re{\anchor{\regbos}\regstar{.}\group{1}{.}\regstar{.}\anchor{\regeos}} \quad \subreg_2 = \re{\anchor{\regbos}\lazystar{.}\group{1}{.}\regstar{.}\anchor{\regeos}} \quad s = \str{ab} $$
Both regexes $\subreg_1$ and $\subreg_2$ have two matches on the string $s$, where the capture group is mapped to either \str{a} or \str{b}, but \str{b} has higher priority than \str{a} for regex $\subreg_1$ while \str{a} has higher priority than \str{b} for regex $\subreg_2$.

However, \jsregexec is \fpnp-complete.
Indeed, JavaScript regex parsing is in \fpnp: JavaScript regex parsing is in OptP, and as Krentel~\cite{krentel_optp} points out, all problems in OptP are in \fpnp.
And since one can recover group boundaries from the result of JavaScript regex parsing in linear time, \jsregexec is also in \fpnp \cite[Thm. 3.2(i)]{krentel_optp}.
\jsregexec is also \fpnp-hard: the metric reduction from \lexsat to JavaScript regex parsing shown in Section~\ref{subsec:optp-hard} is essentially a reduction to \jsregexec, and \lexsat is \fpnp-hard.

\paragraph*{Mechanization}

Our Rocq development covers the following: For the hardness proofs, we mechanize the correctness of the two PSPACE-hardness reductions and of the OptP-hardness reduction; for membership, we mechanize the monotonicity and the bound on the termination measure, bounds on the size of each state, and the correctness of the simulation algorithm. We also prove that these results extend to Warblre's regexes and semantics by composing Linden proofs with the pre-existing Warblre-Linden equivalence. In contrast, we do not formalize the following: the definitions of PSPACE and OptP in terms of Turing machines (we do not translate our algorithms to Turing machine programs), the space cost of algorithms in terms of concrete memory (we bound the size of each action and the number of actions in Rocq, but we do not have a machine model to relate them to), and the complexity of reductions (we prove reductions correct and bound their output sizes in Rocq, but we rely on manual observation to guarantee their complexity).

Our code can be found at \url{https://github.com/LindenRegex/js-regex-matching-complexity}.

\section{Conclusion}
\label{sec:conclusion}

Modern regex matching is very powerful: we define the problems of matching and parsing regular expressions with expanded lower-bounded quantifiers, and show that they are PSPACE- and Opt-complete, respectively. Our main results are mechanized in Rocq against a faithful model of ECMAScript's regular expression matching specification, marking the first time mechanized complexity results have been given for a real-world regex language.

\appendix

\section{Formal constructions for PSPACE-hardness}
\label{app:formal_constr_hardness}

\subsection{With negative lookarounds}
\label{appsubsec:formal_constr_pspace}

Uezato~\cite{uezato_rewblk} proves that regex matching with backreferences and lookarounds is PSPACE-hard by reducing \qbfprime to it.
However, the reduction Uezato describes makes use of duplicate capture groups, and the semantics of regular expressions in~\cite{uezato_rewblk} is not the same as in JavaScript; for instance, the semantics in~\cite{uezato_rewblk} is not a backtracking one, and it allows all iterations of quantifiers to match the empty string.
Therefore, Uezato's result is not directly applicable to JavaScript.
In this section, we adapt this reduction to JavaScript regular expressions and prove its correctness in Rocq.
The reduction we present here is also partly inspired from~\cite{perl-3sat}, where a reduction from 3-SAT to Perl regex matching is presented.

We polynomially reduce \qbfprime to \jsregmatch; this reduction transforms in polynomial time each \qbfprime input formula to a regex and a string, such that the regex finds a match on the string if and only if the original formula is true.
As in~\cite{uezato_rewblk}, the idea is to let the backtracking regex matching algorithm simulate the naïve algorithm for checking the validity of a \qbfprime formula.
In this process, the group map, which remembers the values of each capture group, simulates the environment that maps each variable to $\top$ or $\bot$.

We start with an example.
Let $q_\mathrm{a} = \neg\exists x_1. \exists x_2. \neg\exists x_3. (x_1 \lor x_2) \land (\neg x_2 \lor x_3)$ be a QBF in the input format of \qbfprime.
We want to set up a regex and string of size polynomial in $|q_\mathrm{a}|$ such that $q_\mathrm{a}$ is true if and only if the regex matches the string from the beginning. Our reduction produces the following regex $r_\mathrm{a}$ and the following string $s_\mathrm{a}$:
\begin{align*}
    r_\mathrm{a} &= \re{\neglookahead{\mathnormal{E(x_1) E(x_2)} \neglookahead{\mathnormal{E(x_3)} \noncap{\disjunction{\backref{1}}{\backref{2}}}{;}\noncap{\disjunction{\backref{2}a}{\backref{3}}}{;}}}} \\
    s_\mathrm{a} &= \str{a;a;a;a;a;z} \\
    \text{where } E(x_i) &= \re{\noncap{\disjunction{\group{\mathit{i}}{a}}{a}};}
\end{align*}
Here is a breakdown of the regex. First, the core, $\re{\noncap{\disjunction{\backref{1}}{\backref{2}}}{;}\noncap{\disjunction{\backref{2}a}{\backref{3}}}{;}}$. We assume first that the group map maps each group from 1 to 3 to either \str{a} or \vundef. Under this assumption, the sequence \re{\noncap{\disjunction{\backref{1}}{\backref{2}}};\noncap{\disjunction{\backref{2}a}{\backref{3}}};} matches the string \str{a{;}a{;}} if and only if both subregexes match the string \str{a{;}}. The components of the sequence match if and only if the groups are correctly set up. For example, the fragment \re{\noncap{\disjunction{\backref{2}a}{\backref{3}}};} matches the string \str{a{;}} if and only if group 2 is set to \vundef (false) or group 3 is set to \str{a} (true), hence if and only if the formula $\neg x_2 \lor x_3$ is true given values for each group. Thus, the whole core fragment tests one assignment: it matches the string \str{a{;}a{;}} if and only if the underlying propositional formula of the QBF $q_a$ is true given values for each variable.

Then, the assignment fragment, i.e. the part with the $E(x_i)$ and the lookarounds, forces the backtracking engine to consider all possible assignments.
From now on, we say that a group map $\gm$ \emph{agrees} with an assignment $e$ of variables $x_1, x_2, ..., x_n$ to $\bot$ or $\top$ (an \emph{environment}) when for each $x_i$, $\gm$ maps group $i$ to \str{a} if and only if $e$ maps $x_i$ to $\top$, and $\gm$ has no mapping for group $i$ if and only if $e$ maps $x_i$ to $\bot$.
We say that a group map $\gm$ agrees with an environment $e$ on variables in $X$ if and only if the above condition holds for each $x_i \in X$.
$E(x_3) = \re{\noncap{\disjunction{\group{3}{a}}{a}};}$ always matches the string \str{a{;}}, but has two ways of doing so, either by capturing the \chr{a}, hence setting group $3$ to \str{a}, or by not capturing the \chr{a}, leaving group $3$ undefined.
Therefore, the backtracking engine tries matching $E(x_3)$ followed by the inner fragment on the string \str{a{;}a{;}a{;}} by first setting group $3$ to \str{a} (true), then trying to match the inner fragment on the string \str{a{;}a{;}}.
If this does not succeed, it backtracks, leaving group $3$ set to the \vundef (false), and tries matching the inner fragment again on the string \str{a{;}a{;}}.
Therefore, given a group map and an environment that agree on variables $x_1$ and $x_2$ and do not define variable $x_3$, the regex $\re{\mathnormal{E(x_3)} \noncap{\disjunction{\backref{1}}{\backref{2}}};\noncap{\disjunction{\backref{2}a}{\backref{3}}};}$ matches the string \str{a{;}a{;}a{;}} if and only if the partially quantified boolean formula $\exists x_3. (x_1 \lor x_2) \land (\neg x_2 \lor x_3)$ is true.
Going further outside, the negative lookahead \re{\neglookahead{\mathnormal{E(x_3)} \noncap{\disjunction{\backref{1}}{\backref{2}}};\noncap{\disjunction{\backref{2}a}{\backref{3}}};}} matches the string \str{a{;}a{;}a{;}} if and only if the lookahead regex $E(x_3) \re{\noncap{\disjunction{\backref{1}}{\backref{2}}};\noncap{\disjunction{\backref{2}a}{\backref{3}}};}$ does \emph{not} match \str{a{;}a{;}a{;}}, i.e. if and only if the partially quantified boolean formula $\neg\exists x_3. (x_1 \lor x_2) \land (\neg x_2 \lor x_3)$ is true (note the negation before the existential quantifier).
Then similarly, given a group map and environment that agree on variable $x_1$ and do not define $x_2$ and $x_3$, the regex \re{\mathnormal{E(x_2)} \neglookahead{\mathnormal{E(x_3)} \noncap{\disjunction{\backref{1}}{\backref{2}}};\noncap{\disjunction{\backref{2}a}{\backref{3}}};}} matches the string \str{a{;}a{;}a{;}a{;}} if and only if the partially quantified boolean formula $\exists x_2. \neg\exists x_3. (x_1 \lor x_2) \land (\neg x_2 \lor x_3)$ is true.
Then, the regex \re{\mathnormal{E(x_1) E(x_2)} \neglookahead{\mathnormal{E(x_3)} \noncap{\disjunction{\backref{1}}{\backref{2}}};\noncap{\disjunction{\backref{2}a}{\backref{3}}};}} matches the string \str{a{;}a{;}a{;}a{;}a{;}} with an empty group map if and only if the quantified boolean formula $\exists x_1. \exists x_2. \neg\exists x_3. (x_1 \lor x_2) \land (\neg x_2 \lor x_3)$ is true.
Finally, with the outer negative lookahead, the regex $r_a$ matches the string \str{a{;}a{;}a{;}a{;}a{;}} with an empty group map if and only if the QBF $q_a$ is true. Since the final \chr{z} is not used anywhere during the matching process, and a regex does not need to match the entire string to match the said string, it turns out that the regex $r_a$ matches the string $s_a$ if and only if the QBF $q_a$ is true.

More generally, let $q = Q_1x_1. Q_2x_2. ...Q_nx_n. \varphi(x_1, x_2, ..., x_n)$ be a QBF that respects the input format of \qbfprime, where for all $i$, $Q_i \in \{\exists, \neg\exists\}$ and $\varphi$ is a possibly negated propositional formula in CNF.
Let $\varphi(x_1, x_2, ..., x_n) = C_1 \land C_2 \land ... \land C_m$ or $\varphi(x_1, x_2, ..., x_n) = \neg(C_1 \land C_2 \land ... \land C_m)$, where each of the $C_j$ is a clause (disjunction of literals). Let for all $j$, $C_j = l_{j, 1} \lor l_{j, 2} \lor ... \lor l_{j, v_j}$ where each $l_{j,k}$ is a literal $x_i$ or $\neg x_i$. We define:
\begin{align*}
    R_{\mathrm{lit}}(x_i) &= \re{\backref{\mathnormal{i}}} \quad \text{for all } i \in \{1, ..., n\} \\
    R_{\mathrm{lit}}(\neg x_i) &= \re{\backref{\mathnormal{i}}a} \quad \text{for all } i \in \{1, ..., n\}
\end{align*}
This way, given a group map and environment that agree on variable $x_i$, $R_{\mathrm{lit}}(x_i)$ matches the string \str{a} with the group map if and only if the group map maps $i$ to the string \str{a}, i.e. if and only if variable $x_i$ is true in the environment.
Conversely, $R_{\mathrm{lit}}(\neg x_i)$ matches the string \str{a} if and only if group $i$ is set to \vundef, i.e. if and only if the literal $\neg x_i$ is true under the environment.

We then define:
$$
    R_{\mathrm{cl}}(C_j) = \re{\noncap{\mathnormal{R_{\mathrm{lit}}(l_{j,1})} | \mathnormal{R_{\mathrm{lit}}(l_{j,2})} | ... | \mathnormal{R_{\mathrm{lit}}(l_{j,v_j})}};} \quad \text{for all } j \in \{1, ..., m\}
$$
This way, given a group map and an environment that agree, $R_{\mathrm{cl}}(C_j)$ matches the string \str{a;} with the group map if and only if either of the $R_{\mathrm{lit}}(l_{j, k})$ matches the string \str{a} with that group map, i.e. if and only if either of the literals $l_{j, k}$ is true, i.e. if and only if clause $C_j$ is true under the environment.

We then define:
$$
    R_{\mathrm{conj}}(\varphi) = R_{\mathrm{cl}}(C_1) \reskip R_{\mathrm{cl}}(C_2) \reskip ... \reskip R_{\mathrm{cl}}(C_m)
$$
This way, given a group map and environment that agree, $R_{\mathrm{conj}}(\varphi)$ matches the string \str{a;a;...;a;} (\str{a;} $m$ times) with the group map if and only if each of the $R_{\mathrm{cl}}(C_j)$ matches the string \str{a;} with the group map, i.e. if and only if all of the clauses $C_j$ are true, i.e. if and only if the conjunction $C_1 \land C_2 \land ... \land C_j$ is true under the environment.

We then define:
$$
    R_{n+1} = R_{\mathrm{prop}}(\varphi) = \begin{cases}
        R_{\mathrm{conj}}(\varphi) & \text{if } \varphi(x_1, ..., x_n) = C_1 \land C_2 \land ... \land C_m \\
        \re{\neglookahead{\mathnormal{R_{\mathrm{conj}}}}} & \text{if } \varphi(x_1, ..., x_n) = \neg(C_1 \land C_2 \land ... \land C_m)
    \end{cases}
$$
This way, given a group map and environment that agree:
\begin{itemize}
    \item if $\varphi$ is not negated, then $R_\mathrm{prop}(\varphi)$ matches the string \str{a;a;...;a;} (\str{a;} $m$ times) with the group map if and only if $R_\mathrm{conj}(\varphi)$ matches the same string with the group map, i.e. if and only if $\varphi = C_1 \land C_2 \land ... \land C_m$ is true under the environment, according to the previous explanation,
    \item if $\varphi$ is negated, then $R_\mathrm{prop}(\varphi)$ matches the string \str{a;a;...;a;} (\str{a;} $m$ times) with the group map if and only if $R_\mathrm{conj}(\varphi)$ does \emph{not} match the same string, i.e. if and only if $C_1 \land C_2 \land ... \land C_m$ is false under the environment (according to the previous explanation), i.e. if and only if $\varphi = \neg(C_1 \land C_2 \land ... \land C_m)$ is true under the environment.
\end{itemize}

Then, we define:
\begin{align*}
    E(x_i) &= \re{\noncap{\disjunction{\group{\mathit{i}}{a}}{a}};} \quad \text{for all } i \in \{1, ..., n\} \\
    R_i &= \begin{cases}
        E(x_i) R_{i+1} & \text{if } Q_i = \exists \\
        \re{\neglookahead{\mathnormal{E(x_i)R_{i+1}}}} & \text{if } Q_i = \neg\exists
    \end{cases}
\end{align*}
Let $q_i = Q_ix_i. Q_{i+1}x_{i+1}. ... Q_nx_n. \varphi(x_1, ..., x_n)$ and consider a group map and environment that agree on variables $x_1$ to $x_{i-1}$ and leave variables $x_i$ to $x_n$ undefined. Then, $R_i$ matches the string \str{a;a;...;a;} (\str{a;} $m+(n+1-i)$ times) with the group map if and only if the partially quantified boolean formula $q_i$ is true under the environment. To see this, proceed by descending induction on $i \in \{1, ..., n+1\}$. The case $i = n+1$ corresponds to the characterization of $R_{\mathrm{prop}}(\varphi)$ given at the previous step. Then, given $i \in \{1, ..., n\}$, assuming the property to be true for $i+1$, consider a group map and environment that agree on variables $x_1$ to $x_{i-1}$ and leave variables $x_i$ to $x_n$ undefined:
\begin{itemize}
    \item if $Q_i = \exists$, then $R_i = E(x_i) R_{i+1} = \re{\noncap{\disjunction{\group{i}{a}}{a}};} R_{i+1}$ matches the string \str{a;a;...;a;} (\str{a;} $m+(n+1-i)$ times) under the group map if and only if $R_{i+1}$ matches the string \str{a;...;a;} (\str{a;} $m+(n+1-(i+1))$ times) under the group map with group $i$ set to either \str{a} or \vundef. This is equivalent by the induction hypothesis to $q_{i+1}$ being true under the environment with either $x_{i+1} = \top$ or $x_{i+1} = \bot$, i.e. to $q_i = \exists x_i. q_{i+1}$ being true under the environment;
    \item if $Q_i = \neg\exists$, then $R_i = \re{\neglookahead{\mathnormal{E(x_i) R_{i+1}}}}$ matches the string \str{a;a;...;a;} (\str{a;} $m+(n+1-i)$ times) with the group map if and only if $E(x_i) R_{i+1}$ does \emph{not} match the string \str{a;a;...;a;} (\str{a;} $m+(n+1-i)$ times) with the group map. By the previous point, this is equivalent to $\exists x_i. q_{i+1}$ being \emph{false}, i.e. to $q_i = \neg\exists x_i. q_{i+1}$ being true under the environment.
\end{itemize}

Finally, we define:
\begin{align*}
    r &= R_1 \\
    s &= \str{a;a;...;a;z} \quad \text{(\str{a;} $n+m$ times, followed by \str{z})}
\end{align*}
We reduce the problem of checking whether $q$ is true to the problem of checking whether $r$ matches $s$ from the beginning of the string.

The reduction is done in polynomial time.
Indeed, an amount of work proportional to the size of the underlying variable is done per quantifier to map each $Q_ix_i$ to either $E(x_i)$ or $\re{\neglookahead{\mathnormal{E(x_i) ...}}}$.
Then for each clause, the reduction performs a constant amount of work (appending a `{;}' to the clause-checking regex) plus an amount of work proportional to the size of the variable for each variable (adding a disjunction plus either $\re{\backref{\mathnormal{i}}}$ or $\re{\backref{\mathnormal{i}}a}$ for variable $x_i$).

The proof that the above reduction from \qbfprime to \jsregmatch is correct, meaning that a QBF is true if and only if the corresponding regex matches the corresponding string, was mechanized in the Rocq proof assistant.
It relies on six theorems, each specifying a part of the reduction: one about the $R_{\mathrm{lit}}(l_{j,k})$, one about the $R_{\mathrm{cl}}(C_j)$, one about $R_\mathrm{conj}(\varphi)$, one about $R_{\mathrm{prop}}(\varphi)$, one about $E(x_i)$, and one about the $R_i$.
The final theorem about $r$ is as follows:

\begin{theorem}
    The regex $r$ corresponding to a QBF $q$ matches the string $s$ corresponding to the QBF $q$ if and only if $q$ is true.
\end{theorem}

\subsection{Without negative lookarounds}
\label{appsubsec:formal_constr_pspace_noneglk}

\subsubsection{Encoding negation without negative lookarounds}
\label{subsubsec:negation_poslk}

Previous work~\cite{uezato_rewblk} and our reduction in the previous subsection relies on negative lookarounds to negate regex matching and thus encode negation. Indeed, it seems at first that negative lookarounds are the only feature of JavaScript regexes that allow such a negation of regex matching. However, we find that the use of negative lookarounds is in fact not essential to encode negated existential quantifiers, and that positive lookarounds are enough to encode a limited form of negation; this reduces the subset of JavaScript regex features that make JavaScript regex matching PSPACE-hard.

Indeed, in JavaScript and other languages such as Perl~\cite{perlre} and PCRE2~\cite{pcre2}, (positive) lookarounds are \emph{atomic}, which means that once a match for the lookaround sub-regex is found, if no match is found for the rest of the regex, no further backtracking is performed inside the lookaround. Therefore, if we know that a branch of the lookaround matched, then we know that there was no higher priority branch of the lookaround that could have matched.

In concrete terms, to check whether a regex $r$ does \emph{not} match a string $s = \str{a;...;a;z}$, we check whether the regex $$ \re{\lookahead{\disjunction{\mathnormal{r}}{\mathnormal{R_{\mathrm{cap}}(\gid_\mathrm{z})}}}\mathnormal{R_{\mathrm{chk}}(\gid_\mathrm{z})}} $$ matches the string $s$, where for all $\gid$, $R_{\mathrm{cap}}(\gid) = \re{\regstar{\noncap{a{;}}}\group{\gid}{z}}$ and $R_{\mathrm{chk}}(\gid) = \re{\regstar{\noncap{a{;}}}\backref{\gid}\anchor{\regeos}}$ and $\gid_\mathrm{z}$ is a fresh group name.
If $r$ indeed does not match $s$, then the right branch of the disjunction $R_{\mathrm{cap}}(\gid_\mathrm{z}) = \re{\regstar{\noncap{a{;}}}\group{\gid_\mathrm{z}}{z}}$ matches \str{a;...;a;z}, capturing the final \chr{z} in group $\gid_\mathrm{z}$, the lookaround succeeds and the remaining regex $R_{\mathrm{chk}}(\gid_\mathrm{z}) = \re{\regstar{\noncap{a{;}}}\backref{\gid_\mathrm{z}}\anchor{\regeos}}$ matches the string $s = \str{a;...;a;z}$.
Conversely, if $r$ does match $s$ (it does not capture group $\gid_\mathrm{z}$, since $\gid_\mathrm{z}$ is fresh), then the lookahead succeeds without capturing group $\gid_\mathrm{z}$, and the remaining regex $R_{\mathrm{chk}}(\gid_\mathrm{z})$ does \emph{not} match the string $s$.
Crucially, since lookarounds are atomic in JavaScript, the matching process does \emph{not} attempt to backtrack in the lookaround to try to match the right branch of the disjunction, which would succeed while making group $\gid_\mathrm{z}$ capture the final \chr{z} and make the entire regex match the string $s$.
Therefore, in that case, the entire regex does not match $s$.
All in all, we have that the regex \re{\lookahead{\disjunction{\mathnormal{r}}{\mathnormal{R_{\mathrm{cap}}(\gid_\mathrm{z})}}}\mathnormal{R_{\mathrm{chk}}(\gid_\mathrm{z})}} matches the string $s$ if and only if the regex $r$ does \emph{not} match the string $s$.

\subsubsection{Reducing \qbfprime to JavaScript regex matching without negative lookarounds}
To see how we leverage this principle to reduce \qbfprime to JavaScript regex matching without negative lookarounds, we start with an example.
Consider the QBF $q_\mathrm{b} = \exists x_1. \neg\exists x_2. x_1 \lor x_2$.
We want to build a regex without negative lookarounds and a string such that the regex matches the string from the beginning of the string if and only if $q_\mathrm{b}$ is true.
Our new reduction produces the following regex $r_\mathrm{b}$ and string $s_\mathrm{b}$:
\begin{align*}
    r_\mathrm{b} &= \re{\mathnormal{E(x_1)} \lookahead{\disjunction{\noncap{\mathnormal{E(x_2)} \noncap{\disjunction{\backref{1}}{\backref{2}}}{;}}}{\mathnormal{R}_{\mathrm{cap}}(\texttt{z}_2)}}\mathnormal{R_{\mathrm{chk}}(\texttt{z}_2)}} \\
    s_\mathrm{b} &= \str{a;a;a;z}
\end{align*}
where $E(x_i)$ is defined as in the previous subsection and $R_{\mathrm{cap}}(\gid)$ and $R_{\mathrm{chk}}(\gid)$ are defined as in Appendix \ref{subsubsec:negation_poslk} for all $\gid$.

Here is a breakdown of the regex. As in Appendix \ref{appsubsec:formal_constr_pspace}, the core of the regex, \re{\noncap{\disjunction{\backref{1}}{\backref{2}}}{;}}, matches the string \str{a;z} given values of groups 1 and 2 if and only if the propositional formula $x_1 \lor x_2$ underlying the QBF $q_\mathrm{b}$ is true given the corresponding values of the variables $x_1$ and $x_2$.
Going further outside, as in Appendix \ref{appsubsec:formal_constr_pspace}, given a group map and environment that agree on variable $x_1$ and leave $x_2$ undefined, the regex \re{\mathnormal{E(x_2)} \noncap{\disjunction{\backref{1}}{\backref{2}}}{;}} matches the string \str{a;a;z} with the group map if and only if the partially quantified boolean formula $\exists x_2. x_1 \lor x_2$ is true under the environment.
Then, as explained above, the regex \re{\lookahead{\disjunction{\noncap{\mathnormal{E(x_2)} \noncap{\disjunction{\backref{1}}{\backref{2}}}{;}}}{\mathnormal{R}_{\mathrm{cap}}(\texttt{z}_2)}}\mathnormal{R_{\mathrm{chk}}(\texttt{z}_2)}} matches the string \str{a;a;z} with the group map if and only if the inner regex \re{\mathnormal{E(x_2)} \noncap{\disjunction{\backref{1}}{\backref{2}}}{;}} does \emph{not} match the string \str{a;a;z} with the group map, hence if and only if the partially quantified boolean formula $\exists x_2. x_1 \lor x_2$ is false, i.e. the partially quantified boolean formula $\neg\exists x_2. x_1 \lor x_2$ is true under the environment.
Finally, as in Appendix \ref{appsubsec:formal_constr_pspace}, the regex \re{\mathnormal{E(x_1)} \lookahead{\disjunction{\noncap{\mathnormal{E(x_2)} \noncap{\disjunction{\backref{1}}{\backref{2}}}{;}}}{\mathnormal{R}_{\mathrm{cap}}(\texttt{z}_2)}}\mathnormal{R_{\mathrm{chk}}(\texttt{z}_2)}} matches the string \str{a;a;a;z} if and only if there exists a value of group 1, either \str{a} or \vundef, that makes the regex \re{\lookahead{\disjunction{\noncap{\mathnormal{E(x_2)} \noncap{\disjunction{\backref{1}}{\backref{2}}}{;}}}{\mathnormal{R}_{\mathrm{cap}}(\texttt{z}_2)}}\mathnormal{R_{\mathrm{chk}}(\texttt{z}_2)}} match the string \str{a;a;z}, i.e. if and only if there exists a value of variable $x_1$ making the partially quantified boolean formula $\neg\exists x_2. x_1 \lor x_2$ true, i.e. if and only if the QBF $q_\mathrm{b} = \exists x_1. \neg\exists x_2. x_1 \lor x_2$ is true.

More generally, let $q = Q_1x_1. Q_2x_2. ...Q_nx_n. \varphi(x_1, x_2, ..., x_n)$ be a QBF that respects the input format of \qbfprime, as in Appendix \ref{appsubsec:formal_constr_pspace}, where for all $i$, $Q_i \in \{\exists, \neg\exists\}$ and $\varphi$ is a possibly negated propositional formula in CNF. Let $\varphi(x_1, x_2, ..., x_n) = C_1 \land C_2 \land ... \land C_m$ or $\varphi(x_1, x_2, ..., x_n) = \neg(C_1 \land C_2 \land ... \land C_m)$, where each of the $C_j$ is a clause. We define:
\begin{align*}
    R_{\mathrm{conj}}(\varphi) &= R_{\mathrm{cl}}(C_1) \reskip R_{\mathrm{cl}}(C_2) \reskip ... \reskip R_{\mathrm{cl}}(C_m) \quad \text{as before} \\
    R'_{n+1} = R'_{\mathrm{prop}}(\varphi) &= \begin{cases}
        R_\mathrm{conj}(\varphi) & \text{if } \varphi(x_1, ..., x_n) = C_1 \land C_2 \land ... \land C_m \\
        \re{\lookahead{\disjunction{\mathnormal{R_\mathrm{conj}(\varphi)}}{\mathnormal{R_\mathrm{cap}(\texttt{z}_\mathnormal{n+1})}}}\mathnormal{R_\mathrm{chk}(\texttt{z}_\mathnormal{n+1})}} & \text{if } \varphi(x_1, ..., x_n) = \neg(C_1 \land C_2 \land ... \land C_m)
    \end{cases}
\end{align*}
This way, consider a group map and an environment that agree on variables $x_1, ..., x_n$ such that the group map does not define group $\texttt{z}_{n+1}$; as in Appendix \ref{appsubsec:formal_constr_pspace}, $R_{\mathrm{conj}}(\varphi)$ matches the string \str{a;...;a;z} (\str{a;} $m$ times followed by a \chr{z}) with the group map if and only if the conjunction $C_1 \land C_2 \land ... \land C_m$ is true under the environment. Then:
\begin{itemize}
    \item If $\varphi$ is not negated, then $R'_\mathrm{prop}(\varphi) = R_\mathrm{conj}(\varphi)$ matches the string \str{a;...;a;z} with the group map if and only if $\varphi(x_1, ..., x_n)$ is true under the environment.
    \item If $\varphi$ is negated, then as explained in Appendix \ref{subsubsec:negation_poslk}, the regex $$ R'_\mathrm{prop}(\varphi) = \re{\lookahead{\disjunction{\mathnormal{R_\mathrm{conj}(\varphi)}}{\mathnormal{R_\mathrm{cap}(\texttt{z}_\mathnormal{n+1})}}}\mathnormal{R_\mathrm{chk}(\texttt{z}_\mathnormal{n+1})}} $$ matches the string \str{a;...;a;z} with the group map if and only the inner regex, $R_\mathrm{conj}(\varphi)$, does \emph{not} match the string \str{a;...;a;z} with the group map, hence if and only if $C_1 \land C_2 \land ... \land C_m$ is false, i.e. if and only if $\varphi(x_1, ..., x_n) = \neg(C_1 \land C_2 \land ... \land C_m)$ is true under the environment.
\end{itemize}
We then define:
\begin{align*}
    R_{\mathrm{cap}}(\gid) &= \re{\regstar{\noncap{a{;}}}\group{\gid}{z}} \\
    R_{\mathrm{chk}}(\gid) &= \re{\regstar{\noncap{a{;}}}\backref{\gid}\anchor{\regeos}} \\
    R'_i &= \begin{cases}
        E(x_i) R'_{i+1} & \text{if } Q_i = \exists \\
        \re{\lookahead{\disjunction{\mathnormal{E(x_i)R'_{i+1}}}{\mathnormal{R_{\mathrm{cap}}(\texttt{z}_\mathnormal{i})}}}\mathnormal{R_{\mathrm{chk}}(\texttt{z}_\mathnormal{i})}} & \text{if } Q_i = \neg\exists
    \end{cases}
\end{align*}
Let $q_i = Q_ix_i. Q_{i+1}x_{i+1}. ... Q_nx_n. \varphi(x_1, ..., x_n)$ for all $i \in \{1, ..., n+1\}$. Then, given a group map and an environment that agree on variables $x_1$ to $x_{i-1}$ such that the group map does not define any other groups, $R'_i$ matches the string \str{a;a;...;a;z} (\str{a;} $m+(n+1-i)$ times followed by a \chr{z}) with the group map if and only if $q_i$ is true under the environment. To see this, we proceed by descending induction on $i \in \{1, ..., n+1\}$. The case $i = n+1$ corresponds to the above characterization of $R'_{\mathrm{prop}}(\varphi)$. Then, let $i \in \{1, ..., n\}$, assume the property to be true for $i+1$, and consider a group map and environment that agree on variables $x_1$ to $x_{i-1}$ such that the group map does not define any other groups:
\begin{itemize}
    \item If $Q_i = \exists$, then similarly to Appendix \ref{appsubsec:formal_constr_pspace}, $R'_i = E(x_i) R'_{i+1}$ matches the string \str{a;a;...;a;z} if and only if $R'_{i+1}$ matches the string \str{a;...;a;z} with the group map with group $i$ set to either \str{a} or \vundef, i.e. if and only if $q_{i+1}$ is true under the environment with some value of the variable $x_i$, i.e. if and only if $q_i = \exists x_i. q_{i+1}$ is true under the environment.
    \item If $Q_i = \neg\exists$, then by the previous point, $E(x_i) R'_{i+1}$ matches the string \str{a;a;...;a;z} with the group map if and only if $\exists x_i. q_{i+1}$ is true under the environment.
    Then, as explained in Appendix \ref{subsubsec:negation_poslk}, the regex $$R'_i = \re{\lookahead{\disjunction{\mathnormal{E(x_i)R'_{i+1}}}{\mathnormal{R_{\mathrm{cap}}(\texttt{z}_\mathnormal{i})}}}\mathnormal{R_{\mathrm{chk}}(\texttt{z}_\mathnormal{i})}}$$ matches the string \str{a;a;...;a;z} with the group map if and only if the inner regex $E(x_i) R'_{i+1}$ does not match the string \str{a;a;...;a;z} with the group map, i.e. if and only if $\exists x_i. q_{i+1}$ is false, i.e. if and only if $q_i = \neg\exists x_i. q_{i+1}$ is true under the environment.
\end{itemize}
We finally define:
\begin{align*}
    r' &= R'_1 \\
    s &= \str{a;a;...;a;z} \quad \text{(\str{a;} $n+m$ times, followed by \str{z})}
\end{align*}
We reduce the problem of determining whether $q$ is true to the problem of determining whether $r'$ matches $s$.

The correctness of the reduction was proved in Rocq, with significant proof sharing with the reduction that uses negative lookarounds. Three theorems (plus auxiliary lemmas) were needed: one specifying the behavior of the $R_{\mathrm{cap}}(\texttt{z}_i)$, one specifying the behavior of the $R_{\mathrm{chk}}(\texttt{z}_i)$, and one specifying the behavior of the $R'_i$. The final theorem follows:

\begin{theorem}
    The regex $r'$ corresponding to a QBF $q$ matches the string $s$ corresponding to the QBF $q$ if and only if $q$ is true.
\end{theorem}

\section{Pseudo-code of matching algorithm for PSPACE membership}
\label{app:pspace_membership_pseudocode}

Below is the description of a recursive function \texttt{compute\_result} that computes the result (not a backtracking tree) of JavaScript regex matching given a list of actions, input, group map, matching direction, and a fuel argument bounding the recursion depth. Below, we define a function \texttt{pspace\_algo} that uses \texttt{compute\_result} to check whether a given regex matches a given input string.

\begin{lstlisting}[language=Coq]
compute_result act inp gm dir 0 = Out_of_fuel
compute_result [] inp gm dir (S fuel) = Success (inp, gm)
compute_result (Acheck strcheck :: cont) inp gm dir (S fuel) =
    if (is_strict_suffix inp strcheck dir) then
        compute_result cont inp gm dir fuel
    else NoMatch
compute_result (Aclose gid :: cont) inp gm dir (S fuel) =
    compute_result cont inp (GroupMap.close (idx inp) gid gm) dir fuel
compute_result (Epsilon :: cont) inp gm dir (S fuel) =
    compute_result cont inp gm dir fuel
compute_result (Regex.Character cd :: cont) inp gm dir (S fuel) =
    match read_char rer cd inp dir with
    | Some (c, nextinp) =>
        compute_result cont nextinp gm dir fuel
    | None => NoMatch
    end
compute_result (Disjunction r1 r2 :: cont) inp gm dir (S fuel) =
    match compute_result (r1 :: cont) inp gm dir fuel with
    | Out_of_fuel => Out_of_fuel
    | Success lf => Success lf
    | NoMatch => compute_result (r2 :: cont) inp gm dir fuel
    end
compute_result (Sequence r1 r2 :: cont) inp gm dir (S fuel) =
    compute_result (seq_list r1 r2 dir ++ cont) inp gm dir fuel
compute_result (Quantified greedy (S min) delta r1 :: cont) inp gm dir (S fuel) =
    let gidl := def_groups r1 in
    compute_result (r1 :: Quantified greedy min delta r1 :: cont) inp
        (GroupMap.reset gidl gm) dir fuel
compute_result (Quantified greedy 0 0 r1 :: cont) inp gm dir (S fuel) =
    compute_result cont inp gm dir fuel
compute_result (Quantified true 0 delta r1 :: cont) inp gm dir (S fuel) =
    let gidl := def_groups r1 in
    match
        compute_result
            (r1 :: Acheck inp ::
                Quantified true 0 (delta-1) r1 :: cont)
            inp (GroupMap.reset gidl gm) dir fuel
    with
    | Out_of_fuel => Out_of_fuel
    | Success lf => Success lf
    | NoMatch => compute_result cont inp gm dir fuel
    end
compute_result (Quantified false 0 delta r1 :: cont) inp gm dir (S fuel) =
    let gidl := def_groups r1 in
    match (compute_result cont inp gm dir fuel) with
    | Out_of_fuel => Out_of_fuel
    | Success lf => Success lf
    | NoMatch =>
        compute_result
            (r1 :: Acheck inp :: Quantified false 0 (delta-1) r1 :: cont)
            inp (GroupMap.reset gidl gm) dir fuel
    end
compute_result (Group gid r1 :: cont) inp gm dir (S fuel) =
    compute_result (r1 :: Aclose gid :: cont) inp (GroupMap.open (idx inp) gid gm) dir fuel
compute_result (Lookaround lk r1 :: cont) inp gm dir (S fuel) =
    match (compute_result [r1] inp gm (lk_dir lk) fuel) with
    | Out_of_fuel => Out_of_fuel
    | Success (_,gmlk) =>
        match (positivity lk) with
        | false => NoMatch
        | true => compute_result cont inp gmlk dir fuel
        end
    | NoMatch =>
        match (positivity lk) with
        | true => NoMatch
        | false => compute_result cont inp gm dir fuel
        end
    end
compute_result (Anchor a :: cont) inp gm dir (S fuel) =
    if anchor_satisfied rer a inp then
        compute_result cont inp gm dir fuel
    else NoMatch
compute_result (Backreference gid :: cont) inp gm dir (S fuel) =
    match read_backref rer gm gid inp dir with
    | Some (br_str, nextinp) =>
        compute_result cont nextinp gm dir fuel
    | None => NoMatch
    end

pspace_algo (r:regex) (s:string) =
    let init_fuel := fuel r input(s) in
    match compute_result [r] (input(s)) GroupMap.empty forward init_fuel with
    | Success _ => Some true
    | NoMatch => Some false
    | OutOfFuel => None (* we proved this does not happen *)
    end
\end{lstlisting}
The {\fuel} function is described in Section~\ref{sssec:extension_lookarounds}.
The {\computeresult} function is very similar to the functional semantics of~\cite{linden_popl26}, except that instead of computing the backtracking tree, it explores branches of the backtracking tree in a depth-first search order (starting with the highest priority branches).

\section{Full description of PSPACE-membership metric}
\label{app:pspace_membership_metrics}

The algorithm described in Appendix \ref{app:pspace_membership_pseudocode} is similar to the functional semantics of~\cite{linden_popl26}.
Like for the functional semantics of~\cite{linden_popl26}, the termination of $\computeresult$ is not clear a priori, since the list of actions grows and shrinks as the function recurses.
Both therefore take a fuel argument, which decreases by one for each recursive call.
This fuel argument acts as a bound on the number of nested recursions that the function can use at any given time.
Barri\`ere et al. \cite{linden_popl26} show that there exists an amount of fuel for any list of actions, input, group map and matching direction (a \emph{semantic state}) such that the functional semantics called with these arguments and this fuel terminates without running out of fuel.
This fuel would (very likely) also ensure the termination of our algorithm without running out of fuel.
However, this fuel can be exponential in the size of the AST of the regex in the presence of nested quantifiers, such as with the sequence of regexes \re{\regstar{\regepsilon}}, \re{\regstar{\regstar{\regepsilon}}}, \re{\regstar{\regstar{\regstar{\regepsilon}}}}, etc.

In Section~\ref{subsec:functional_termination}, we present a more complex, but much more tight fuel function \texttt{fuel} that ensures the termination of $\computeresult$ without running out of fuel.
This new fuel is polynomial in the regex and input sizes.
In Section~\ref{subsec:state_size}, we additionally prove that all the semantic states involved in the computation of the matching result have a size bounded by a polynomial in the regex and input sizes.

\subsection{Showing termination of the functional semantics of~\cite{linden_popl26} with a polynomial fuel}
\label{subsec:functional_termination}

To prove termination of our algorithm with our new polynomial fuel, we prove that at each step of our algorithm, when there is a recursive call, the fuel associated with the semantic state of the next call is strictly smaller than the one associated to the original state. Namely, we prove:

\begin{theorem}
    For all semantic states $s_1$ and $s_2$, if $s_2$ is one of the next semantic states of $s_1$, then $\fuel(s_2) < \fuel(s_1)$.
    \label{thm:fuel_decrease}
\end{theorem}
This is the key result of our membership proof. This theorem is proved in the next section.

From this theorem, we deduce:
\begin{theorem}
    For all $\actions, \inp, \gm, \dir, f$ such that $f > \fuel(\actions, \inp, \gm, \dir)$, we have\\ $\computeresult~l~\inp~\gm~\dir~f <> \outoffuel$.
    \label{thm:termination}
\end{theorem}
\begin{proof}
    This is a direct consequence of Theorem \ref{thm:fuel_decrease}. A proof of this theorem that includes the arguments of Theorem \ref{thm:fuel_decrease} is mechanized.
\end{proof}

Finally, we show that for any regex and string, the fuel of the initial state is polynomial in the sizes of the regex and string. Using $\gmempty$ to denote the initial group map where every group is set to $\vundef$, and $\inpofidx{s}{i}$ to denote the input $s$ at position $i$:
\begin{theorem}
    For all regex $\subreg$ and string $s$, \\$\fuel([\areg{\subreg}], \inpofidx{s}{0}, \gmempty, \forward) \le (1 + |s|)\times(\regsize{\subreg} + \regsize{\subreg}^2)$.
    \label{thm:polyfuel}
\end{theorem}
\begin{proof}
    This theorem follows from the definition of $\fuel$ in Sections~\ref{sssec:fuel_definition} and \ref{sssec:extension_lookarounds}. This proof is mechanized.
\end{proof}

We will now define the $\fuel$ function and prove Theorem \ref{thm:fuel_decrease}. We first provide a formula for the $\fuel$ function that assumes that there are no lookarounds, then extend it in Section~\ref{sssec:extension_lookarounds}.

\subsubsection{Fuel definition}
\label{sssec:fuel_definition}

Our fuel definition does not need the group map.
To define our termination measure, we first divide the list of actions into \emph{chunks}, sublists of actions separated by the \acheck{} actions. In the semantics in~\cite{linden_popl26}, chunks are created when unrolling non-lower-bounded quantifiers, and erased when successfully passing a progress check.

We define three functions, $\fuelfirst$, $\fuelmiddle$ and $\fuellast$ for different chunks. We define $\fuel(\actions, \inp)$ as follows: given $\actions = c_1 \app \acheck{\inp_1} :: c_2 \app ... \app \acheck{\inp_{n-1}} :: c_n$,
\begin{align*}
    \fuel(\actions, \inp) &= \fuelfirst(c_1, \prog(\actions, \inp)) + \sum_{k=2}^{n-1} (\fuelmiddle(c_k)+1) + \fuellast(c_n, \inp, \prog(\actions, \inp)) \\
    \text{where } \prog(\actions, \inp) &= \begin{cases}
    1 & \text{if } \forall \inp', \acheck{\inp'} \in \actions \Rightarrow \inp < \inp' \\
    0 & \text{otherwise}
    \end{cases}
\end{align*}
Given $\actions = [a_1; ...; a_n]$,
\begin{align*}
    \fuelfirst(\actions, p) &= \begin{cases}
        1 + \sum_{k=2}^n |a_k| & \text{if $a_1$ is a non-lower-bounded quantifier and $p = 1$} \\
        \sum_{k=1}^n |a_k| & \text{otherwise}
    \end{cases} \\
    \fuelmiddle(\actions) &= 1 + \sum_{k=2}^n |a_k| \\
    \fuellast(\actions, \inp, p) &= \left(\sum_{k=1}^n |a_k|\right) \times (|\inp| + p)
\end{align*}
and $|\aclose{\gid}| = 1$ for all $\gid$.

We say that a check input $\inp$ is \emph{smaller} than a check input $\inp'$, denoted $\inp < \inp'$, when $\inp$ represents the same string as $\inp'$ but has a strictly further position in the string than $\inp'$.

\subsubsection{A strictly decreasing fuel (Theorem \ref{thm:fuel_decrease})}

In this section, we focus on explaining why this function is a strictly decreasing termination measure (Theorem \ref{thm:fuel_decrease}). We focus on high-level arguments. For any details, the proof has been fully mechanized.

\paragraph*{Sums of sizes}

At the heart of our fuel formula, we sum the sizes of actions in multiple places.
The main idea is that, for any regex feature except non-lower-bounded quantifiers, the sum of the sizes of actions in the action list provides a polynomial, strictly decreasing termination measure.
Indeed, while the list of actions both grows (e.g. when unrolling a sequence) and shrinks (e.g. when consuming a character) as we explore the tree, at most steps, the sum of the sizes of all actions in the list strictly decreases.
For instance, in the case of the sequence, the list grows but the sum of sizes decreases by one.

While this simple observation fails in the presence of quantifiers, e.g. when unfolding one iteration of a star, it still explains why a large part of our definitions consists in summing sizes of actions.
Everything that is not a sum of sizes is to deal with non-lower-bounded quantifiers.

\paragraph*{Multiplication by $|\inp|+1$ for the last chunk}

In the $\fuellast$ function, we multiply the sum of the sizes of each action by the length of the remaining input plus one (we ignore cases where $\prog(\actions, \inp)$ is zero for now; this will be explained later).

One key observation is that in JavaScript, each optional iteration has to consume at least one character in order for the subsequent progress check to pass. A consequence is that non-lower-bounded quantifiers can only do as many iterations as the remaining length of the string plus one.
Barri\`ere et al. \cite{linden_popl26} used that argument, to define a fuel function where we sum the sizes of each action except for quantifiers, for which the size is multiplied by the remaining length of input.
This fuel function is proved to be strictly decreasing, but it is not polynomial in general. Consider for instance the case of the list of actions $[\re{\regstar{\regstar{\subreg}}}]$ for a regex $\subreg$. For a given input $\inp$,~\cite{linden_popl26} gives it fuel $\lVert \re{\regstar{\regstar{\subreg}}} \rVert_\inp = \lVert \re{\regstar{\subreg}} \rVert_\inp \times |\inp| = \lVert \subreg \rVert_\inp \times |\inp| \times |\inp|$. Each level of nested star multiplies the entire fuel by the length of the input, which can make the fuel grow exponentially with the size of input.

In order to be more precise in the case of nested stars, we present a new key observation: \textbf{Nested Iterations Observation:} each subregex can be matched at most $|\inp|+1$ times, regardless of how many nested quantifiers it is inside of. For instance, when matching \re{\regstar{\regstar{\subreg}}}, it is possible for both the inner and the outer star to be iterated $|\inp|+1$ times. However, it is not possible to have $|\inp|+1$ iterations of the outer star, each of which doing $|\inp|+1$ iterations of the inner star. If one of the outer star iterations iterates the inner star $|\inp|+1$ times, then we consumed at least $|\inp|$ characters in the string, and the outer star cannot iterate anymore.

That is why in the definition of $\fuellast$, in order to be polynomial, we multiply the whole sum of sizes by $|\inp|+1$ at most. If a regex with nested stars ends up in that chunk, the whole size will be multiplied only once by the input length.

While this is more precise, this also makes the formula more complex: we have to treat other chunks than the last one in a different way.

\paragraph*{Handling other chunks}

When we unroll a quantifier, for instance when the list of actions $[\re{\regstar{\regstar{\subreg}}}]$ becomes $[\re{\regstar{\subreg}} :: \acheck{\inp} :: \re{\regstar{\regstar{\subreg}}}]$, we end up with multiple chunks, increasing the size of our list. Our formula cannot measure quantifiers in these new chunks the same way that it measured quantifiers in the last chunk, otherwise the fuel would increase.

We again leverage the nested iterations observation. In this example, $\re{\regstar{\subreg}}$ can be iterated $n_1$ times and $\re{\regstar{\regstar{\subreg}}}$ can be iterated $n_2$ times, but we should ensure that $n_1 + n_2 \leq |\inp|+1$ so that the common subregex $\subreg$ is iterated at most $|\inp|+1$ times.

We are trying to compute a bound on the maximum number of steps that executing this list of actions could take. In our example, since $\re{\regstar{\subreg}}$ has a smaller size than $\re{\regstar{\regstar{\subreg}}}$, the worst case happens when iterating $\re{\regstar}$ only once ($n_1 = 1$), and iterating $\re{\regstar{\regstar{\subreg}}}$ $|\inp|$ times ($n_2 = |\inp|$). In that example for the worst case, we need $n_1$ to be at least $1$, otherwise the \acheck{} action will not succeed.

To generalize this idea, we prove the following invariants:
\begin{invariant}
    In any reachable list of actions, each \acheck{} action is followed by a non-lower-bounded quantifier.
\end{invariant}
\begin{invariant}
    Except for the last chunk, the size of any chunk is smaller than the size of the next quantified regex (just after the check at the end of the chunk).
    \label{inv:chunk_size_lt}
\end{invariant}
With these invariants, we know that every chunk has a size smaller than the quantifier that follows it, and that the last chunk has a bigger size than any other chunk. By multiplying the last chunk by $|\inp|+1$ and not the others, we ensure:
\begin{enumerate}
    \item that we enforce the nested iterations observation, not giving too much fuel to nested stars, and
    \item that by multiplying by $|\inp|+1$ the size of the biggest chunk (and not the others), we have enough fuel to handle the worst-case scenario where the bigger regex ($\re{\regstar{\regstar{\subreg}}}$ in our example) is iterated more times.
\end{enumerate}

\paragraph*{Measuring quantifiers outside of the last chunk}
Using that same principle (iterating the last chunk is more costly than iterating other part of the list of actions), when a non-lower-bounded quantifier immediately follows an \acheck{} action, the worst-case scenario consists in skipping that quantifier, allowing more iterations of the last chunk.
This is why the first action of each middle chunk is measured as 1.

Similarly, when the first action of the first chunk is a quantifier, we may skip it and measure it as 1, but only in the case where skipping the quantifier does not prevent us from passing the subsequent checks (i.e., only when all checks in the list pass). In cases like $[\re{\regstar{\subreg}} :: \acheck{\inp} :: \re{\regstar{\regstar{\subreg}}}]$ at input $\inp$, in the worst case, \re{\regstar{\subreg}} is iterated once in order to pass the $\acheck{\inp}$ action, and then be able to iterate the last chunk as many times as possible.

\paragraph*{Using progress}
In our formula, we multiply the last chunk by either $|\inp|$ or $|\inp|+1$ depending on the value of $\prog(\actions, \inp)$, which encodes whether or not we have made enough progress in the string to pass all checks in the list of actions.
When at least one check does not pass ($\prog(\actions, \inp) = 0$), it means that we are currently iterating a quantifier, but we can't iterate it further unless we read one character during this current iteration and pass the \acheck{} action successfully.
This means that the quantifier after the check can be iterated at most $|\inp|$ times, not $|\inp|+1$, since we only execute them after one character has been read.
On the contrary, if all checks pass ($\prog(\actions, \inp) = 1$), then it is possible to pass all \acheck{} actions without consuming any more characters, leading to at most $|\inp|+1$ iterations of the last chunk.

To mechanize this argument, we had to prove the following invariant:
\begin{invariant}
    All check inputs in a reachable list of actions are ordered from more restrictive to less restrictive.
\end{invariant}

\paragraph*{Example}

To see how this formula is strictly decreasing even when dealing with quantifiers, consider matching the regex \re{\regstar{a}} on the string \str{aa}.

At some point, we reach the list of actions $\actions_1 = [a :: \acheck{\str{aa}} :: \re{\regstar{a}}]$ with input $\inp_1 = \str{aa}$. The fuel corresponding to this list of actions and input is $1 + 1 + (4 \times 2) = 10$ (since $\prog(\actions_1, \inp_1) = 0$). As we read the first letter, we move on to the lists $\actions_2 = [\acheck{\str{aa}} :: \re{\regstar{a}}]$, then $\actions_3 = [\re{\regstar{a}}]$, at input $\inp_2 = \inp_3 = \str{a}$. Their fuel is respectively $1 + (4 \times (1 + 1)) = 9$ and $4 \times (1 + 1) = 8$ ($\prog(\actions_2, \inp_2) = \prog(\actions_3, \inp_3) = 1$).

On the longest branch of the tree, the one which iterates the star again, we move on to the list of actions $\actions_4 = [a :: \acheck{\str{a}} :: \re{\regstar{a}}]$, with input $\inp_4 = \str{a}$. This pair has fuel $1 + 1 + (4 \times 1) = 6$ (because $\prog(\actions_4, \inp_4) = 0$).

We can see that at each step, the fuel function has strictly decreased.

\subsubsection{Extension of the \fuel~function to lookarounds}
\label{sssec:extension_lookarounds}

In order to take lookarounds into account, we add enough fuel to match the lookarounds from any position in the string, and adapt our fuel definition to support both matching directions.

We define the function $\reglkfuel$ as follows:
\begin{align*}
    \reglkfuel(s, \re{\regepsilon}) = \reglkfuel(s, \re{\cd}) &= 0 \\
    \reglkfuel(s, \re{\anc}) = \reglkfuel(s, \re{\backref{\gid}}) &= 0 \\
    \reglkfuel(s, \re{\group{\gid}{\subreg}}) = \reglkfuel(s, \re{\quant{\subreg}{\rmin}{\Delta}{\greedy}}) &= \reglkfuel(s, \subreg) \\
    \reglkfuel(s, \re{\disjunction{\subreg_1}{\subreg_2}}) = \reglkfuel(s, \re{\sequence{\subreg_1}{\subreg_2}}) &= \max(\reglkfuel(s, \subreg_1), \reglkfuel(s, \subreg_2)) \\
    \reglkfuel(s, \re{\lookaround{\lk}{\subreg}}) &= (1 + |s|) \times \regsize{\subreg} + \reglkfuel(s, \subreg)
\end{align*}

We lift this definition to actions by defining the function $\actlkfuel$ as $$ \actlkfuel(s, \actions) = \max_{\areg{\subreg} \in \actions} \reglkfuel(s, \subreg) $$

We finally (re)define $\fuel(\actions, \inp, \dir)$ as follows: given $\actions = c_1 \app \acheck{\inp_1} :: c_2 \app ... \app \acheck{\inp_{n-1}} :: c_n$,
\begin{align*}
    \fuel(\actions, \inp, \dir) &= \fuelfirst(c_1, \prog(\actions, \inp, \dir)) + \sum_{k=2}^{n-1} (\fuelmiddle(c_k)+1) \\
    &+ \fuellast(c_n, \inp, \dir, \prog(\actions, \inp, \dir)) + \actlkfuel(\strof{\inp}, \actions) \\
    \text{where } \prog(\actions, \inp, \dir) &= \begin{cases}
    1 & \text{if } \forall \inp', \acheck{\inp'} \in \actions \Rightarrow \inp <_\dir \inp' \\
    0 & \text{otherwise}
    \end{cases}
\end{align*}
Given $\actions = [a_1; ...; a_n]$,
\begin{align*}
    \fuelfirst(\actions, p) &= \begin{cases}
        1 + \sum_{k=2}^n |a_k| & \text{if $a_1$ is a non-lower-bounded quantifier and $p = 1$} \\
        \sum_{k=1}^n |a_k| & \text{otherwise}
    \end{cases} \\
    \fuelmiddle(\actions) &= 1 + \sum_{k=2}^n |a_k| \\
    \fuellast(\actions, \inp, \dir, p) &= \left(\sum_{k=1}^n |a_k|\right) \times (|\inp|_\dir + p)
\end{align*}
We write $|\inp|_\dir$ for the remaining length of input $\inp$ in direction $\dir$.
We say that a check input $\inp$ is \emph{smaller} than a check input $\inp'$ in direction $\dir$, denoted $\inp <_\dir \inp'$, where $\inp$ represents the same string as $\inp'$ but has a position in the string that is strictly further in direction $\dir$ than $\inp'$.

As in~\cite{linden_popl26}, we consider the worst case for the strings of each lookaround, which means taking as remaining length of the input the length of the entire input string.

\subsection{Bounding the size of the semantic states used in recursive calls}
\label{subsec:state_size}
We now polynomially bound the memory footprint of the semantic states used in recursive calls by a polynomial in the size of the input string $s$ and regex $\subreg$.

Such a semantic state contains a list of actions, an input, a group map and a matching direction.

The memory footprint of the input is bounded by $|s|$.
The group map contains at most as many pairs of indices into the string as there are groups in the regex, and there are at most $\regsize{\subreg}$ such groups.
Therefore, the memory footprint of the group map is $\mathrm{O}(\regsize{\subreg} \times \log{|s|})$.
The matching direction has a constant memory footprint.

We now polynomially bound the \emph{sizes} of the lists of actions reachable from $\subreg$ and $s$.
The key idea to bound these sizes is to use invariant \ref{inv:chunk_size_lt} about the lists of actions, which states that each chunk has a size less than that of the quantified regex following the next check action.
This entails that each chunk has a size less than the next one, which allows bounding the sum of the sizes of the chunks by $m(m+1)/2$ where $m$ is the size of the last chunk.
Finally, the size of the last chunk is less than or equal to the size of the original regex.
This is intuitively because the last chunk is the result of unrolling the original regex without considering duplicates incurred by unrolling a quantified regex, and that cannot increase the size of the last chunk.

We end up with the following theorem, which was proven in Rocq:
\begin{theorem}
    For any list of actions $\actions$ reachable from a regex $\subreg$, $$ |\actions| \le |\subreg| + |\subreg|(|\subreg|+1)/2 $$
    where $|\actions|$ is the sum of the sizes of each action of $\actions$ and $|\acheck{\inp}| = 1$.
\end{theorem}
However, the size of the list of actions may be an underestimate of its memory footprint, mainly due to us supporting quantified regexes with an unbounded maximum number of repetitions.
Indeed, the size we define for the regex $\re{\quant{\subreg}{0}{\rmax}{\greedy}}$ is $3 + |\re{\subreg}|$, discarding entirely the memory footprint of $\rmax$. $\rmax$ may not even appear in the maximum numbers of repetitions encountered in the original regex.
To solve this problem, we notice the following: the bound above remains a bound on the number of actions in the list of actions, and we can easily bound the memory footprint of each action by either the memory footprint of the original regex (for any action other than an \acheck{} action) or by $\mathrm{O}(\log{|s|})$ (for \acheck{} actions, we remember the index into the input string).

Therefore, the memory footprint of any list of actions reachable from a regex $\subreg$ is $\mathrm{O}((|\subreg| + \log{|s|}) \times [|\subreg| + |\subreg|(|\subreg|+1)/2])$.

This yields a bound on the memory footprint of the semantic states of the functional semantics which is polynomial in the sizes of the initial regex and input string.

\subsection{Putting everything together}
\label{subsec:pspace_algo}

Putting everything together, we obtain the following theorem, which is \emph{not} mechanized:
\begin{theorem}
    For any regex $\subreg$ and string $s$, executing $\pspacealgo{\subreg}{s}$ requires polynomial space (in the size of $s$ and the expanded size of $\subreg$).
    \label{thm:polyspace}
\end{theorem}
\begin{proof}
We know that during the execution, we never use more stack frames than the initial fuel. From Theorem \ref{thm:polyfuel}, we know that this number is polynomial in the size of $s$ and the expanded size of $\subreg$. From Section~\ref{subsec:state_size}, we know that each of these stack frames themselves require a polynomial amount of space.
\end{proof}

\paragraph*{Correctness of the algorithm} We prove:
\begin{theorem}
    Let $\subreg$ be a regex and $s$ be a string.
    Let $t$ be the backtracking tree such that $\istree{[\areg{\subreg}]}{\inpofidx{s}{0}}{\gmempty}{\forward}{t}$.
    Then $\pspacealgo{\subreg}{s}$ returns the first leaf of $t$.
    \label{thm:pspace_linden}
\end{theorem}
\begin{proof}
    We prove the following generalized result:
    Let $\actions$ be a list of actions, $\inp$ an input, $\gm$ a group map, $\dir$ a matching direction and $f \in \mathbb{N}$.
    Assume that $\computeresult~\actions~\inp~\gm~\dir~f$ terminates without running out of fuel.
    Let $t$ be the backtracking tree such that $\istree{\actions}{\inp}{\gm}{\dir}{t}$.
    Then $\computeresult~\actions~\inp~\gm~\dir~f$ returns the first leaf of $t$.

    This is proven by induction on $f$.
\end{proof}

This entails:
\begin{theorem}
    For any regex $\subreg$ and string $s$, $\pspacealgo{\subreg}{s}$ returns the boolean result of the instance of \jsregmatch with $\subreg$ and $s$, according to the ECMAScript 2023 specification.
    \label{thm:pspace_correctness}
\end{theorem}
\begin{proof}
Compose Theorem \ref{thm:pspace_linden} with the proof of faithfulness of the semantics in~\cite{linden_popl26} (Section 4).
\end{proof}

We conclude from Theorems \ref{thm:polyspace} and \ref{thm:pspace_correctness}:
\begin{theorem}
    JavaScript regex matching with expanded lower-bounded quantifiers is in PSPACE.
    \label{thm:pspace}
\end{theorem}

\section{Formal construction for OptP-hardness}
\label{app:formal_constr_optp}

\subsection{Principle of reduction from CNF \lexsat}

In the absence of negated existential quantifiers, the two reductions presented in Appendices \ref{appsubsec:formal_constr_pspace} and \ref{appsubsec:formal_constr_pspace_noneglk} are identical. This reduction is essentially also a reduction from CNF \lexsat to JavaScript regex \emph{parsing} without lookarounds. Indeed, consider for instance the following propositional formula in CNF with three variables $x_1$, $x_2$ and $x_3$:
$$ \varphi_\mathrm{c} = (x_1 \lor x_2) \land (\neg x_2 \lor x_3) $$
The corresponding QBF $q_\mathrm{c} = \exists x_1. \exists x_2. \exists x_3. \varphi_\mathrm{c}$ is encoded by the reductions of Appendix \ref{app:formal_constr_hardness} into the problem of deciding whether the following regex $r_\mathrm{c}$ matches the following string $s_\mathrm{c}$:
\begin{align*}
    r_\mathrm{c} &= E(x_1) E(x_2) E(x_3) R_{\mathrm{prop}}(\varphi_\mathrm{c}) \\
    s_\mathrm{c} &= \str{a;a;a;a;a;z}
\end{align*}
where $R_{\mathrm{prop}}(\varphi_\mathrm{c}) = \re{\noncap{\disjunction{\backref{1}}{\backref{2}}}{;}\noncap{\disjunction{\backref{2}a}{\backref{3}}}{;}}$. As explained earlier, each regex $E(x_i)$ either sets the capture group $i$ to \str{a} (representing $\top$) or leaves it undefined (representing $\bot$), and after all capture groups are set up, $R_{\mathrm{prop}}(\varphi_\mathrm{c})$ checks whether the resulting environment satisfies the propositional formula $\varphi_\mathrm{c}$. What also makes this reduction a reduction from CNF \lexsat to JavaScript regex parsing without lookarounds is the following observations:
\begin{itemize}
    \item The backtracking semantics of JavaScript regexes gives priority to first setting group $1$ to \str{a} (left branch of the disjunction) before leaving it undefined (right branch of the disjunction). Within a setup of group 1, the backtracking semantics gives priority to first setting group $2$ to \str{a} before leaving it undefined, and similarly for group 3. As a consequence, the backtracking semantics of JavaScript specifies that environments are considered with an order of priority corresponding to descending lexicographic order.
    \item When the first match of a regex on a string is found, the parsing process returns the list of backtracking choices made, from which the values of the capture groups (the group map) can be deduced.
    Returning the group map is equivalent in our case to returning the highest priority environment that makes the propositional formula true.
    \item In the absence of negated existential quantifiers in the QBF, the reduction does not involve negative or positive lookarounds in the regex.
\end{itemize}
Therefore, to solve our instance of CNF \lexsat, one can simply try matching $r_\mathrm{c}$ on $s_\mathrm{c}$ from the beginning of the string and straightforwardly convert the resulting group map into an environment mapping variables to $\top$ or $\bot$ (or return $0$ if $r_\mathrm{c}$ does not match $s_\mathrm{c}$).

More generally, let $\varphi(x_1, x_2, ..., x_n) = C_1 \land C_2 \land ... \land C_m$ be a propositional formula in CNF, where each of the $C_j$ is a clause. We reduce the resulting instance of CNF \lexsat to the problem of matching the following regex $r$ on the following string $s$:
\begin{align*}
    R_{\mathrm{prop}}(\varphi) &= R_{\mathrm{cl}}(C_1) \reskip R_{\mathrm{cl}}(C_2) \reskip ... \reskip R_{\mathrm{cl}}(C_m) \quad \text{as before} \\
    r &= E(x_1) \reskip E(x_2) \reskip ... \reskip E(x_n) R_{\mathrm{prop}} \\
    s &= \str{a;a;...;a;z} \quad \text{(\str{a;} $n+m$ times, followed by \str{z})}
\end{align*}

\subsection{Proof of validity of reduction}

We now prove that the above reduction from CNF \lexsat to JavaScript regex parsing without lookarounds is valid, which means that given an instance of CNF \lexsat, one can recover the result of this instance from the result of the corresponding instance of regex parsing.
This proof is mechanized in the Rocq proof assistant.

We set $R_\mathrm{conj}(\varphi) = R_{n+1} = R_\mathrm{prop}(\varphi)$ and $R_i = E(x_i) \reskip E(x_{i+1}) \reskip ... \reskip E(x_n) R_\mathrm{prop}(\varphi)$ for all $i \in \{1, ..., n+1\}$; these notations match the reduction presented in Appendix \ref{appsubsec:formal_constr_pspace}.
We say that a group map $\gm$ is \emph{well-formed} when for all $i \in \{1, ..., n\}$, either $\gm(i)$ is undefined or $\gm(i) = (2(i-1), 2(i-1)+1)$ (corresponding to the $i$-th \chr{a} in the input string).
We then define, for all well-formed group maps $\gm$ and $\gm'$:
\begin{multline*}
    \gm' >_\mathrm{lex} \gm \Leftrightarrow \exists i \in \{1, ..., n\}, \\ \gm'(i) = (2(i-1), 2(i-1)+1) \land \gm(i) = \bot \land \forall j \in \{1, ..., i-1\}, \gm'(j) = \gm(j)
\end{multline*}
where $\gm(i) = \bot$ means that $\gm(i)$ is undefined; this is a strict partial order on the set of well-formed group maps. Finally, we define the following function $\mathcal{G}$ that computes the group map of the first leaf corresponding to a list of actions, input, group map and matching direction:
\begin{definition}
    Let $\actions$ be a list of actions, $inp$ be an input, $\gm$ be a group map and $\dir$ be a matching direction. We define $\mathcal{G}(\actions, inp, \gm, \dir)$ as follows. Let $\treecont$ be the backtracking tree such that $\istree{\actions}{inp}{\gm}{\dir}{\treecont}$. Then:
    \begin{itemize}
        \item if $\treecont$ has a leaf, and its first leaf is $(inp', \gm')$, we define $\mathcal{G}(\actions, inp, \gm, \dir) = \gm'$,
        \item if $\treecont$ has no leaf, $\mathcal{G}(\actions, inp, \gm, \dir)$ is undefined (denoted $\mathcal{G}(\actions, inp, \gm, \dir) = \bot$).
    \end{itemize}
\end{definition}
This function is similar to the function $\mathcal{L}_0$ in~\cite{linden_popl26} that computes the first leaf of a backtracking tree for an input.

We now prove the following invariant:

\begin{theorem}
    Let $i \in \{1, ..., n+1\}$. Let $inp = \inpofidx{s}{2(i-1)}$ be the input in string $s$ at position $2(i-1)$ (right before the $i$-th \chr{a}).
    Let $\gm_0$ be a well-formed group map such that $\gm_0(i), \gm_0(i+1), ..., \gm_0(n)$ are undefined.
    \begin{enumerate}
        \item \label{itm:hasleaf} Assume that $\mathcal{G}([R_i], inp, \gm_0, \forward) = \gm$ for some $\gm$. Then:
        \begin{enumerate}
            \item \label{itm:sat_phi} $\gm$ satisfies $\varphi$,
            \item \label{itm:same_prefix} $\gm$ and $\gm_0$ coincide on indices $1, 2, ..., i-1$,
            \item \label{itm:greater_unsat} for any well-formed $\gm'$ such that $\gm' >_\mathrm{lex} gm$ and $\gm'$ and $\gm$ coincide on indices $1, 2, ..., i-1$, $\gm'$ does not satisfy $\varphi$.
        \end{enumerate}
        \item \label{itm:hasnoleaf} Assume that $\mathcal{G}([R_i], inp, \gm_0, \forward)$ is undefined. Then for any well-formed $\gm$ that coincides with $\gm_0$ on indices $1, 2, ..., i-1$, $\gm$ does not satisfy $\varphi$.
    \end{enumerate}
\end{theorem}

\begin{proof}[Proof (sketch)]
We proceed by descending induction on $i \in \{1, ..., n+1\}$.

First, let $i = n+1$. Let $inp = \inpofidx{s}{2n}$ and $\gm_0$ be a well-formed group map.
\begin{itemize}
    \item Assume that $\mathcal{G}([R_{n+1}], inp, \gm_0, \forward) = \gm$ for some $\gm$. Then statements \ref{itm:sat_phi} and \ref{itm:same_prefix} are consequences of the theorem about $R_\mathrm{conj}(\varphi)$ in Appendix \ref{appsubsec:formal_constr_pspace}, while statement \ref{itm:greater_unsat} has an impossible condition (a group map $\gm'$ cannot both coincide with $\gm$ on all indices and be lexicographically greater than $\gm$).
    \item Statement \ref{itm:hasnoleaf} follows from the fact that the statement ``$\gm$ satisfies $\varphi$'' does not depend on indices other than $1, 2, ..., n$: essentially, the only $\gm$ that coincides with $\gm_0$ on indices $1, 2, ..., n$ is $\gm_0$ itself.
\end{itemize}

Now, let $i \in \{1, ..., n\}$ and assume that the invariant holds for $i+1$. Let $inp_i = \inpofidx{s}{2(i-1)}$ and $\gm_0$ be a well-formed group map such that $\gm_0(i), \gm_0(i+1), ..., \gm_0(n)$ are undefined.

To prove statement \ref{itm:hasleaf}, assume that $\mathcal{G}([R_i], inp_i, \gm_0, \forward) = \gm$ for some $\gm$. Then, letting $inp_{i+1} = \inpofidx{s}{2i}$, one the following is true:
\begin{itemize}
    \item $\mathcal{G}([R_{i+1}], inp_{i+1}, \gm_0[i \mapsto (2(i-1), 2(i-1)+1)], \forward) = \gm$.
    \item $\mathcal{G}([R_{i+1}], inp_{i+1}, \gm_0[i \mapsto (2(i-1), 2(i-1)+1)], \forward) = \bot$ and $\mathcal{G}([R_{i+1}], inp_{i+1}, \gm_0, \forward) = \gm$.
\end{itemize}
We then reason by case analysis:
\begin{itemize}
    \item In the first case, the induction hypothesis directly implies statement \ref{itm:sat_phi}, and it almost directly implies statement \ref{itm:same_prefix}. For statement \ref{itm:greater_unsat}, notice that the condition implies $\gm' >_\mathrm{lex} \gm$ and that $\gm'$ and $\gm$ coincide on indices $1, 2, ..., i=i+1-1$: the induction hypothesis applies.
    \item In the second case, the induction hypothesis directly implies statement \ref{itm:sat_phi}, and it almost directly implies statement \ref{itm:same_prefix}. For statement \ref{itm:greater_unsat}, the condition states that $\gm' >_\mathrm{lex} \gm$ and that $\gm'$ and $\gm$ coincide on indices $1, 2, ..., i-1$. Perform case analysis on whether $\gm'(i) = \gm(i)$; notice that $\gm(i) = \bot$ by statement \ref{itm:same_prefix} of the induction hypothesis.
    \begin{itemize}
        \item If $\gm'(i) = \bot$, then $\gm'$ and $\gm$ coincide on indices $1, 2, ..., i = i+1-1$ and statement \ref{itm:greater_unsat} of the induction hypothesis applies.
        \item If $\gm'(i) = (2(i-1), 2(i-1)+1)$, then statement \ref{itm:hasnoleaf} of the induction hypothesis applies, because $\mathcal{G}([R_{i+1}], inp_{i+1}, \gm_0[i \mapsto (2(i-1), 2(i-1)+1)], \forward) = \bot$.
    \end{itemize}
\end{itemize}

To prove statement \ref{itm:hasnoleaf}, assume that $\mathcal{G}([R_i], inp_i, \gm_0, \forward) = \bot$. Let $inp_{i+1} = \inpofidx{s}{2i}$. Then we have $\mathcal{G}([R_{i+1}], inp_{i+1}, \gm_0[i \mapsto (2(i-1), 2(i-1)+1)], \forward) = \mathcal{G}([R_{i+1}], inp_{i+1}, \gm_0, \forward) = \bot$. For any well-formed group map $\gm$ that coincides with $\gm_0$ on indices $1, 2, ..., i-1$:
\begin{itemize}
    \item if $\gm(i) = \bot$, then the induction hypothesis applies with group map $\gm_0$,
    \item if $\gm(i) = (2(i-1), 2(i-1)+1)$, then the induction hypothesis applies with group map $\gm_0[i \mapsto (2(i-1), 2(i-1)+1)]$.
\end{itemize}
\end{proof}

This invariant implies the following theorem:
\begin{theorem}
    \mbox{}
    \begin{itemize}
        \item Assume that $\mathcal{G}([r = R_1], \inpofidx{s}{0}, \gmempty, \forward) = \gm$ for some $\gm$. Then $\gm$ satisfies $\varphi$, and for any well-formed group map $\gm'$ such that $\gm' >_\mathrm{lex} \gm$, $\gm'$ does not satisfy $\varphi$.
        \item Assume that $\mathcal{G}([r = R_1], \inpofidx{s}{0}, \gmempty, \forward) = \bot$. Then for any well-formed group map $\gm$, $\gm$ does not satisfy $\varphi$.
    \end{itemize}
\end{theorem}

This theorem in turn proves the validity of the reduction:
\begin{theorem}
    The result of the instance $\varphi$ of CNF \lexsat is equal to the transformation of the result of matching the regex $r'$ on the string $s$ into an assignment of propositional variables to $\top$ or $\bot$.
\end{theorem}
The transformation from a group map $\gm$ consists in the assignment that maps each variable $x_i$ to either $\top$ if $\gm(i)$ is defined or $\bot$ otherwise.

\section{Pseudo-code of OptP algorithm}
\label{app:optp_pseudocode}

We provide an algorithm \optpoutput that, given a regex without lookarounds and a string, as well as a list of choices (left or right), tries matching the regex on the string by following the choices in the list of choices instead of performing the full backtracking.
It then returns the list of choices preceded with a 1 (\tttrue) if matching succeeds or a 0 (\ttfalse) if matching fails.
This algorithm is similar to the PSPACE algorithm presented in Appendix \ref{app:pspace_membership_pseudocode} and to the functional semantics of~\cite{linden_popl26}, and is presented below.

\begin{lstlisting}[language=Coq]
Record config := Cfg { cfg_act: actions; cfg_inp: input; cfg_gm: group_map }.

Inductive step := SLeaf (lf: leaf) | SFail | SGo (c: config) | SBranch (hi lo: config).

optp_step dir (Cfg [] inp gm) = SLeaf (inp, gm)
optp_step dir (Cfg (Acheck strcheck :: cont) inp gm) =
  if is_strict_suffix inp strcheck dir then SGo (Cfg cont inp gm) else SFail
optp_step dir (Cfg (Aclose gid :: cont) inp gm) =
  SGo (Cfg cont inp (GroupMap.close (idx inp) gid gm))
optp_step dir (Cfg (Epsilon :: cont) inp gm) = SGo (Cfg cont inp gm)
optp_step dir (Cfg (Regex.Character cd :: cont) inp gm) =
  match read_char rer cd inp dir with
  | Some (_, nextinp) => SGo (Cfg cont nextinp gm)
  | None => SFail
  end
optp_step dir (Cfg (Disjunction r1 r2 :: cont) inp gm) =
  SBranch (Cfg (r1 :: cont) inp gm) (Cfg (r2 :: cont) inp gm)
optp_step dir (Cfg (Sequence r1 r2 :: cont) inp gm) =
  SGo (Cfg (seq_list r1 r2 dir ++ cont) inp gm)
optp_step dir (Cfg (Quantified greedy (S min) delta r1 :: cont) inp gm) =
  SGo (Cfg (r1 :: Quantified greedy min delta r1 :: cont) inp
         (GroupMap.reset (def_groups r1) gm))
optp_step dir (Cfg (Quantified greedy 0 0 r1 :: cont) inp gm) =
  SGo (Cfg cont inp gm)
optp_step dir (Cfg (Quantified greedy 0 (delta + 1) r1 :: cont) inp gm) =
  let iter := Cfg (r1 :: Acheck inp :: Quantified greedy 0 delta r1 :: cont) in
  let skip := Cfg cont inp gm in
  if greedy then SBranch iter skip else SBranch skip iter
optp_step dir (Cfg (Group gid r1 :: cont) inp gm) =
  SGo (Cfg (r1 :: Aclose gid :: cont) inp (GroupMap.open (idx inp) gid gm))
optp_step dir (Cfg (Anchor a :: cont) inp gm) =
  if anchor_satisfied rer a inp then SGo (Cfg cont inp gm) else SFail
optp_step dir (Cfg (Backreference gid :: cont) inp gm) =
  match read_backref rer gm gid inp dir with
  | Some (_, nextinp) => SGo (Cfg cont nextinp gm)
  | None => SFail
  end

optp_algo dir (c: config) ([]: list bool) = Out_of_fuel
optp_algo dir (c: config) ((b :: cs): list bool) =
  match optp_step dir c with
  | SLeaf lf => Success lf
  | SFail => NoMatch
  | SGo c' => optp_algo dir c' cs
  | SBranch hi lo => optp_algo dir (if b then hi else lo) cs
  end

optp_run r inp (cs: list bool) = optp_algo forward (Cfg [r] inp GroupMap.empty) cs

matched (Success _) = true
matched _ = false

optp_output r inp (cs: list bool) = matched (optp_run r inp cs) :: cs
\end{lstlisting}

The \ttconfig{} record above stores a semantic state, which is the combination of a list of actions, an input and a group map.
In the absence of lookarounds, the direction remains constant throughout matching, and hence does not need to be included in the semantic state.
The \optpstep{} function above computes the next step(s) corresponding to a semantic state.
There are four types of next steps, modeled by the \texttt{step} inductive: reaching a leaf, failing to find a match, going to a next configuration without branching, or branching to one of two next configurations.
Given a configuration \config{} and a list of choices \choices, the \optpalgo function above steps from the configuration \config{} by following the list of choices \choices, then returns the corresponding matching result.
Given a regex \ttsubreg, an input \ttinp{} and a list of choices \choices, the \optprun{} function above computes the result of matching \ttsubreg{} on \ttinp{} while following the choices in \choices.
The \matched{} function above computes the bit that allows making all unsuccessful lists of choices precede all successful lists of choices in the ordering of the parsing results.

The \optpoutput function above is meant to be called with a list of choices whose length is equal to $\nbits(\subreg, \texttt{inp}) = (1 + |\texttt{inp}|)\regsize{\subreg} + 1$, where $|\texttt{inp}|$ is the number of remaining characters in the input \texttt{inp}.
This ensures that \optpalgo does not run out of choices to make, and that all considered lists of choices have equal length.

When there is no choice to be made, the choice at the head of the list of choices is ignored and popped.
When there is a choice to be made, namely when matching a disjunction or a non-lower-bounded quantifier, instead of trying both choices, the algorithm uses the choice at the head of the list of choices to decide which choice to make.

\section{Proof of OptP membership of JavaScript regex parsing without lookarounds}

We now show why this algorithm shows that JavaScript regex parsing without lookarounds is in OptP.
We need to prove two things: that this algorithm runs in polynomial time given a list of choices, and that the minimum result of all the branches corresponds to the result of JavaScript regex parsing.

\begin{theorem}
    Given a regex $\subreg$ without lookarounds, a string $s$ and a list of choices of length $\nbits(\subreg, s)$, \optpalgo runs in polynomial time in the size of $s$ and the expanded size of $\subreg$.
    \label{thm:optp_polytime}
\end{theorem}
\begin{proof}
    Given a list of choices of length $\nbits(\subreg, s)$, this algorithm never branches, and it performs a polynomial amount of work per recursive call. Moreover, there are no more than $\fuel~\subreg~\inpofidx{s}{0}$ recursive calls made, and $\fuel~\subreg~\inpofidx{s}{0}$ is polynomial in $\subreg$ and $s$ by Theorem \ref{thm:polyfuel}. Therefore, the algorithm runs in polynomial time given a list of choices.
\end{proof}

\begin{theorem}
    Given a regex $\subreg$ without lookarounds and a string $s$, the maximum result of all calls to \optpalgo with a list of choices of length $\nbits(\subreg, s)$ corresponds to the result of JavaScript regex matching as defined in the ECMAScript 2023 specification.
    \label{thm:optp_correctness}
\end{theorem}
\begin{proof}
    The proof of this theorem was mechanized in the Rocq proof assistant.

    The correctness of this algorithm relies on the following two points:
    \begin{itemize}
        \item No matter the list of choices given to \optpalgo, as long as its length is equal to $\nbits(\subreg, s)$, the result is never \outoffuel. This means that we explore all the branches when ranging over all possible lists of choices.
        \item We order the results by lexicographic order.
        This lexicographic order corresponds to the lexicographic order on the pair of a bit that indicates whether the match succeeded and the list of choices.
        As such:
        \begin{itemize}
            \item if there is no match of $\subreg$ on $s$, no result will start with a 1, hence the lexicographically lowest result of the algorithm is 0 followed by as many ones as needed, which is what is required by the parsing problem,
            \item if there is a match of $\subreg$ on $s$, the lexicographically highest result of the algorithm is 1 followed by the lexicographically highest list of choices that leads to a match of $\subreg$ on $s$, which is what is required by the parsing problem.
        \end{itemize}
    \end{itemize}
\end{proof}

Combining Theorems \ref{thm:optp_polytime} and \ref{thm:optp_correctness}, we obtain:
\begin{theorem}
    JavaScript regex parsing without lookarounds is in OptP.
    \label{thm:optp}
\end{theorem}

\section{Connecting our proofs to the specification text of ECMAScript}
\label{sec:warblre_connect}

There exists two mechanized semantics of JavaScript regular expressions in Rocq: the Warblre mechanization of~\cite{warblre_icfp}, and the backtracking tree semantics of~\cite{linden_popl26}.
The first one is focused on faithfulness and auditability: it is the result of manually translating line-by-line each line of pseudocode from the ECMAScript specification into its Rocq equivalent.
On the other hand, the backtracking tree semantics is designed for succinctness and practicality: it is much easier to conduct proofs using the inductive principle that arises from backtracking tree semantics, than with the matcher functions of Warblre.

Previously, the two semantics had been proved to be equivalent~\cite{linden_popl26}, meaning that the leftmost branch of a backtracking tree is always equal to the result specified by Warblre.
There is however one difference: the backtracking tree semantics uses a datatype to represent regexes that is slightly more expressive than the one used by both Warblre and the ECMAScript specification.
For instance, it is possible to define backtracking tree semantics for regexes that contain badly indexed capture groups, unconstrained duplicated capture groups, or double negations in character classes.
Consequently, for each of our hardness results, we must also prove that the regexes we generate fall within the subset of regexes that can also be expressed as well-formed ECMAScript regexes.

We have formally verified that property for each of our hardness reductions, by redefining the functions so that they target Warblre regexes, proving that these regexes are well-formed (according to the \textit{EarlyErrors} check of ECMAScript), and proving that they are the direct equivalent of the regexes for which we proved the reduction correct.
We compose all of these results, and obtain theorems that are stated directly using Warblre regexes and Warblre semantics.

\section{On unanchored vs. anchored matching}
\label{app:unanchored_anchored}

All of our results are proved for \emph{anchored} matching, where given a regex $r$ and a string $s$, we want to find a match of $r$ in $s$ that starts at the beginning of $s$.
This is the behavior of JavaScript regex matching with the \emph{sticky} (\chr{y}) flag, and contrasts with the default behavior of JavaScript regex matching, where we want to find a match of $r$ that starts anywhere in $s$, with priority given to earlier starting positions in $s$.
We made the choice of dealing with anchored matching, because the core of both the Linden \cite{linden_popl26} and Warblre \cite{warblre_icfp} specifications deal with anchored matching.

To extend our results to unanchored matching, we observe the following.
For hardness results, we notice that our reductions from PCNF-QBF and \lexsat can be modified straightforwardly to target unanchored matching: indeed, in the absence of newline characters in the string (which is the case in our reductions), finding an anchored match of a regex $r$ on a string $s$ is equivalent to finding an unanchored match of $\re{\anchor{\regbos}\subreg}$ on $s$, where $\regbos$ is the beginning-of-string (or beginning-of-line, if the multiline flag \chr{m} is on) anchor.
For membership results, we notice that unanchored matching reduces to anchored matching: finding an unanchored match of a regex $r$ on a string $s$ reduces to finding an anchored match of the regex $\re{\lazystar{\charset{\cdcomplement}}\group{0}{\subreg}}$, where $\re{\charset{\cdcomplement}}$ is a regex that matches all characters, including newline characters.
Indeed, the lazy star gives priority to exiting the quantifier instead of iterating it again, so backtracking semantics will first try matching $r$ from position 0 in the string, then from position 1, then 2, etc. The capture group 0 allows to retrieve the range effectively matched by $r$.

\section{AI Disclosure}
In deciding whether and where to apply Generative AI tools, we distinguished core theorems, where the key arguments are novel observations (the PSPACE hardness and membership proofs against the inductive backtracking tree semantics~\cite{linden_popl26}) from sanity-checking theorems, where the proof is routine and the value of mechanization is only to provide stronger confidence (connecting \qbfprime to QBF, reusing the PSPACE-hardness construction and PSPACE-membership proof arguments for OptP, and extending arguments to cover the Warblre interpreted semantics as explained in Appendix~\ref{sec:warblre_connect}).
For core theorems, we completed reductions and proofs on paper first, then mechanized them in Rocq by hand, with no AI support. For sanity-checking theorems, we used Generative AI to generate Rocq proofs matching paper proof sketches, then manually reviewed the proofs and documented top-level statements.
Generative AI was not used in the authoring of this paper.
To facilitate auditing, our code supplement includes a patch file \texttt{ai.patch} summarizing all modifications that were made with the assistance of generative AI tools.



\bibliography{main}

\end{document}

%% file: macros.tex
\makeatletter
\newcommand\tightparagraph{\def\@toclevel{4}%
  \@startsection{paragraph}{4}{\parindent}%
  {-.15\baselineskip \@minus -.05\baselineskip}%
  {-3.5\p@}%
  {\ACM@NRadjust{\@parfont}}}
\makeatother

\newcommand{\semspace}{\vspace{1.1\baselineskip}}
\newcommand{\newpremise}{\qquad}

\newcommand{\ruledef}[2]{\!{\textsc{\footnotesize#2}\phantomsection\label{rule:#1:#2}}}
\newcommand{\ruleref}[2]{\hyperref[rule:#1:#2]{\textsc{#2}}}

\renewcommand{\autoref}[1]{\Cref{#1}}
\newcommand{\appendixref}[1]{\hyperref[#1]{Appendix~\ref*{#1}}}

\newcommand{\dir}{\mathit{d}}
\newcommand{\forward}{\rightarrow}

\newcommand{\istree}[5]{(#1,#2,#3,#4)\mathbin{\Downarrow}#5}

\newcommand{\ancsymb}{%
  \vcenter{\hbox{\oalign{
        \noalign{\kern-0.7ex}\hfil\rule{.4pt}{1.5ex}\hfil\cr
        \noalign{\kern-0.7ex}$\smile$\cr
        \noalign{\kern-1.8ex}\hfil{\scalebox{.8}{$-$}}\hfil\cr
        \noalign{\kern-1.5ex}\hfil{\scalebox{.5}{$\circ$}}\hfil\cr
      }}}}

\newcommand{\regbos}{\text{\textasciicircum}}
\newcommand{\regeos}{\text{\$}}

\newcommand{\inp}{\mathit{i}}

\newcommand{\gm}{\mathit{gm}}
\newcommand{\gid}{\mathit{g}}

\newcommand{\cont}{\mathit{l}}
\newcommand{\actions}{\mathit{l}}
\newcommand{\cd}{\mathit{cd}}
\newcommand{\anc}{\mathit{a}}
\newcommand{\rchar}{\mathit{c}}
\newcommand{\subreg}{\mathit{r}}
\newcommand{\greedy}{\mathit{p}}

\newcommand{\lk}{\mathit{lk}}

\newcommand{\rmin}{\mathit{min}}
\newcommand{\rmax}{\mathit{max}}
\newcommand{\lf}{\mathit{lf}}

\newcommand{\treecont}{\mathit{t}}
\newcommand{\treeiter}{\mathit{t_{iter}}}

\newcommand{\str}[1]{\texttt{"#1"}}
\newcommand{\chr}[1]{\texttt{#1}}

\newcommand{\cdcomplement}{\hat{~}}

\newcommand{\reskip}{\mskip\medmuskip}
\newcommand{\resskip}{\mskip.5\thinmuskip}

\newcommand{\defineRegexSyntax}{%
\newcommand{\regepsilon}{\varepsilon}%
\newcommand{\regchar}[1]{##1}%
\newcommand{\charset}[1]{[##1]}%
\newcommand{\cdrange}[2]{##1\mathord{-}##2}%
\newcommand{\disjunction}[2]{##1|##2}%
\newcommand{\sequence}[2]{##1\reskip##2}%
\newcommand{\quant}[4]{##1\{##2,##3,##4\}}%
\newcommand{\anchor}[1]{##1}%
\newcommand{\backref}[1]{\backslash##1}%
\newcommand{\lookaround}[2]{(?{##1}\reskip##2)}%
\newcommand{\group}[2]{(_{##1}\resskip##2)}%
\newcommand{\noncap}[1]{\langle##1\rangle}%
\newcommand{\respecial}[1]{\mathord{\mathsf{##1}}}
\newcommand{\regqm}[1]{##1\resskip\respecial{?}}%
\newcommand{\regstar}[1]{##1\resskip\respecial{*}}%
\newcommand{\regplus}[1]{##1\resskip\respecial{+}}%
\newcommand{\lazystar}[1]{##1\resskip\respecial{*?}}%
\newcommand{\lazyplus}[1]{##1\resskip\respecial{+?}}%
\newcommand{\lazyqm}[1]{##1\resskip\respecial{??}}%
\newcommand{\regdot}{\cdot}%
\newcommand{\regall}{\odot}%
\newcommand{\esc}[1]{\backslash\respecial{##1}}%
\newcommand{\posahead}{\respecial{=}}%
\newcommand{\negahead}{\respecial{!}}%
\newcommand{\posbehind}{\respecial{<=}}%
\newcommand{\negbehind}{\respecial{<!}}%
\newcommand{\lookahead}[1]{\lookaround{\posahead}{##1}}%
\newcommand{\neglookahead}[1]{\lookaround{\negahead}{##1}}%
\newcommand{\lookbehind}[1]{\lookaround{\posbehind}{##1}}%
\newcommand{\neglookbehind}[1]{\lookaround{\negbehind}{##1}}%
}

\newcommand{\re}[1]{\ensuremath{{\defineRegexSyntax\mathrm{#1}}}}

\newcommand{\savere}[2]{%
  \expandafter\newsavebox\csname re:#1\endcsname%
  \sbox{\csname re:#1\endcsname}{{\small#2}}%
}
\newcommand{\usere}[1]{\usebox{\csname re:#1\endcsname}}

\newcommand{\ctor}[1]{\textsf{\textup{#1}}}
\newcommand{\fn}[1]{\textsf{\textup{#1}}}

\newcommand{\ctorone}[2]{\expandafter\ifx\expandafter\relax
  \detokenize{#2}\relax\ctor{#1}\else\ctor{#1}~#2\fi}
\newcommand{\ctortwo}[3]{\expandafter\ifx\expandafter\relax
  \detokenize{#2}\relax\ctor{#1}\else\ctor{#1}~#2~#3\fi}
\newcommand{\ctorthree}[4]{\expandafter\ifx\expandafter\relax
  \detokenize{#2}\relax\ctor{#1}\else\ctor{#1}~#2~#3~#4\fi}

\newcommand{\acheck}[1]{\ctorone{Acheck}{#1}}
\newcommand{\aclose}[1]{\ctorone{Aclose}{#1}}
\newcommand{\areg}[1]{#1}

\newcommand{\idx}[1]{\fn{idx}(#1)}
\newcommand{\strof}[1]{\fn{str}(#1)}

\newcommand{\defgroups}[1]{\mathcal{G}(#1)}

\newcommand{\gmclose}[3]{\fn{GM\textsubscript{close}}(#1,#2,#3)}

\newcommand{\gmreset}[2]{\fn{GM\textsubscript{reset}}(#1,#2)}
\newcommand{\gmempty}{\fn{GM\textsubscript{$\emptyset$}}}
\newcommand{\gmclosenoarg}{\fn{GM\textsubscript{close}}}
\newcommand{\gmresetnoarg}{\fn{GM\textsubscript{reset}}}

\newcommand{\treematch}{\ctor{Match}}

\newcommand{\treeclose}[2]{\ctortwo{Close}{#1}{#2}}
\newcommand{\treereset}[2]{\ctortwo{Reset}{#1}{#2}}

\newcommand{\treechoice}[2]{\ctortwo{Choice}{#1}{#2}}

\newcommand{\fuelf}[2]{\lvert\lvert#1\rvert\rvert^{#2}}

\newcommand{\bytecodeinstr}[1]{\texttt{#1}}
\newcommand{\bcone}[2]{\expandafter\ifx\expandafter\relax
  \detokenize{#2}\relax\bytecodeinstr{#1}\else\bytecodeinstr{#1}~#2\fi}
\newcommand{\bctwo}[3]{\expandafter\ifx\expandafter\relax
  \detokenize{#2}\relax\bytecodeinstr{#1}\else\bytecodeinstr{#1}~#2~#3\fi}

\newcommand{\app}{\mathbin{++}}

\definecolor{ex1}{RGB}{201,241,255} 
\definecolor{ex2}{RGB}{255,241,180} 
\definecolor{ex3}{RGB}{198,248,205} 
\definecolor{ex4}{RGB}{255,230,251} 

\newcommand{\inpofidx}[2]{\operatorname{\mathcal{I}}(#1, #2)}

\newcommand{\regsize}[1]{|#1|}
\newcommand{\strsize}[1]{|#1|}
\newcommand{\jsregmatch}{\ensuremath{\mathrm{\textsc{JSRegMatch}}}\xspace}
\newcommand{\jsregexec}{\ensuremath{\mathrm{\textsc{JSRegExec}}}\xspace}
\newcommand{\qbfprime}{\ensuremath{\mathrm{QBF}_\exists}\xspace}
\newcommand{\fpnp}{\ensuremath{\mathrm{FP}^\mathrm{NP}}\xspace}
\newcommand{\lexsat}{\ensuremath{\mathrm{\textsc{Lexicographic SAT}}}\xspace}
\newcommand{\pspacealgo}[2]{\texttt{pspace\_algo}(#1,#2)}
\newcommand{\computeresult}{\texttt{compute\_result}\xspace}
\newcommand{\fuel}{\texttt{fuel}}

\newcommand{\outoffuel}{\texttt{OutOfFuel}}
\newcommand{\optpalgo}{\texttt{optp\_algo}\xspace}

\newcommand{\optpoutput}{\texttt{optp\_output}\xspace}
\newcommand{\optpstep}{\texttt{optp\_step}}
\newcommand{\optprun}{\texttt{optp\_run}}
\newcommand{\matched}{\texttt{matched}}
\newcommand{\ttconfig}{\texttt{config}}
\newcommand{\config}{\texttt{c}}
\newcommand{\choices}{\texttt{cs}}
\newcommand{\ttfalse}{\texttt{false}}
\newcommand{\tttrue}{\texttt{true}}
\newcommand{\ttsubreg}{\texttt{r}}
\newcommand{\ttinp}{\texttt{inp}}
\newcommand{\nbits}{n_b}
\newcommand{\framesize}{s_f}
\newcommand{\fuelfirst}{\operatorname{fuel\_first}}
\newcommand{\fuelmiddle}{\operatorname{fuel\_middle}}
\newcommand{\fuellast}{\operatorname{fuel\_last}}
\newcommand{\prog}{\operatorname{prog}}
\newcommand{\actlkfuel}{\operatorname{act\_lk\_fuel}}
\newcommand{\reglkfuel}{\operatorname{reg\_lk\_fuel}}

\newcommand{\vundef}{\textbf{Undef}\xspace}

\newcommand{\mathcolorbox}[2]{\colorbox{#1}{$\displaystyle #2$}}

\pgfdeclarelayer{background}
\pgfdeclarelayer{foreground}
\pgfsetlayers{background,foreground}